\documentclass[review]{elsarticle}

\usepackage{lineno,hyperref}
\usepackage{graphicx}      % include this line if your document contains figures
\usepackage{siunitx}
\usepackage{amsmath,cancel}
\usepackage{array}
\usepackage{tikz}
\usepackage{mathdots}
\usepackage{units}
\usepackage{yhmath}
\usepackage{cancel}
\usepackage{color}
\usepackage{siunitx}
\usepackage{array}
\usepackage{multirow}
\usepackage{amssymb}
\usepackage{gensymb}
\usepackage{tabularx}
\usepackage{booktabs}
\usepackage{pgfplots}
\usepackage{amsthm}
\usepackage{siunitx}
\usepackage{soul}
\usepackage{adjustbox}
\pgfplotsset{compat=newest}
\usetikzlibrary{plotmarks}
\usetikzlibrary{arrows.meta}
\usepgfplotslibrary{patchplots}
\usepackage{grffile}
\usepackage{amsmath}
\usetikzlibrary{fadings}
\usetikzlibrary{patterns}
\newtheorem{thm}{Theorem}
\newtheorem{rem}{Remark}
\newtheorem{pro}{Proposition}

\usepackage{url}
\usepackage[font=footnotesize,labelfont=bf]{caption}
\usepackage[font=footnotesize,labelfont=bf]{subcaption}
\modulolinenumbers[5]

\journal{Journal of Control Engineering Practice}

\begin{document}

\begin{frontmatter}

\title{Continuous Reset Element:\\ \small{Transient and Steady-state Analysis for Precision Motion Systems} }

\author{Nima Karbasizadeh, S. Hassan HosseinNia\fnref{myfootnote}}
\address{Department
	of Precision and Microsystem Engineering, Delft University of Technology, Delft,
	The Netherlands}
\fntext[myfootnote]{Corresponding Author}




\begin{abstract}
This paper addresses the main goal of using reset control in precision motion control systems, breaking of the well-known ``Waterbed effect''. A new architecture for reset elements will be introduced which has a continuous output signal as opposed to conventional reset elements. A steady-state precision study is presented, showing the steady-state precision is preserved while the peak of sensitivity is reduced. The architecture is then used for a ``Constant in Gain Lead in Phase'' (CgLp) element and a numerical analysis on transient response shows a significant improvement in transient response. It is shown that by following the presented guideline for tuning, settling time can be reduced and at the same time a no-overshoot step response can be achieved.  A practical example is presented to verify the results and also to show that the proposed element can achieve a complex-order behaviour. 
\end{abstract}


\end{frontmatter}

\linenumbers

\section{Introduction}
Waterbed effect limits the performance of the linear control systems~\cite{bode1945network}. Almost every researcher in the field of control engineering has encountered this fundamental limitation. One can come up with different mathematical interpretations for it, however, most definitely, its practical effect is more important, especially for high-tech industrial applications such as precision motion control. One can interpret this effect by putting transient and steady-state response of the system on two sides of this infamous waterbed, which implicates that by improving one, you are sacrificing the other.\\
Reset control systems, first proposed by Clegg in~\cite{clegg1958nonlinear}, are proving themselves as alternatives for linear control systems as they showed potential to outperform linear control systems by breaking waterbed effect limitation. Clegg proposed an integrator whose output will reset to zero whenever its input crosses zero. It was later established that based on Describing Function (DF) analysis, such an action will reduce the phase lag of the integrator by $52^\circ$. Although this already breaks the Bode's gain-phase relation for linear control systems, there are concerns while using Clegg's Integrator (CI) in practice, namely, the accuracy of DF approximation, limit-cycle, etc.\\
In order to address the drawbacks and exploiting the benefits, the idea was later extended to more sophisticated elements such as ``First-Order Reset Element''~\cite{horowitz1975non,krishnan1974synthesis} and ``Second-Order Reset Element''~\cite{hazeleger2016second} or using Clegg's integrator in form of PI+CI~\cite{banos2007definition} or resetting the state to a fraction of its current value, known as partial resetting~\cite{beker2004fundamental}. Reset control has also recently been used to approximate the complex-order filters~\cite{valerio2019reset,saikumar2019constant}. Advantage of using reset control over linear control has been shown in many studies especially in precision motion control ~\cite{BISOFFI202037,banos2011reset,karbasizadeh2020benefiting,beker2004fundamental,dastjerdi2021frequency,Nesic2011stability,wu2006reset,zheng2000experimental,chen2001analysis,karbasizadeh2021fractional}. However, these studies are mostly focused on solving one problem. For example they either improve transient~\cite{Guo2011Optimal} or steady-state response of the system while paying little or no attention to the other.\\
One of the recent studies introduces a new reset element called ``Constant-in-Gain, Lead-in-Phase'' (CgLp) element which is proposed based on the loop-shaping concept~\cite{saikumar2019constant}. DF analysis of this element shows that it can provide broadband phase lead while maintaining a constant gain. Such an element is used in the literature to replace some part of the differentiation action in PID controllers as it will help improve the precision of the system according to loop-shaping concept~\cite{karbasizadeh2021fractional,karbasizadeh2020benefiting,dastjerdi2021frequency,saikumar2019constant}.\\
In~\cite{karbasizadeh2020benefiting,karbasizadeh2021fractional}, it is suggested that DF analysis for reset control systems can be inaccurate as it neglects the higher-order harmonics created in response of reset control systems. These studies also suggest that suppressing higher-order harmonics can improve the steady-state precision of the system.\\
One of the benefits of providing phase lead through CgLp is improving the transient response properties of the system, as it is shown that it reduces the overshoot and settling time of the system. However, the way to achieve this goal is not only through phase compensation around cross-over frequency. It is shown in~\cite{cai2020optimal} that since reset control systems are nonlinear systems, the sequence of elements in control loop affects the output of the system. It was shown that when the lead elements are placed before reset element, it can improve the overshoot of the system. However, no systematic approach is proposed there for further improving the transient response. In~\cite{ZHAO201927}, it is shown that by changing the resetting condition of reset element to reset based on its input and its derivative, overshoot limitation in linear control, systems can be overcome. This limitation has also been broken using the same technique in another hybrid control system called ``Hybrid Integrator Gain System" (HIGS)~\cite{vanden2020hybrid}. However, in these studies the effect of such an action on steady-state performance of the system is not addressed.  \\
Another important common properties of all reset elements in the literature is the discontinuity of the output signal. This properties is a cause for presence of high-frequency content in the signals and subsequent practical issues~\cite{karbasizadeh2021fractional}. Continuous time implementation as opposed to discrete time implementation of reset control and also soft resetting were introduced in the literature to mitigate this problem to some extent~\cite{teel2022continuous,le2021passive}. However, this paper proposes an approach which can also used in discrete time.  \\
The main contribution of this paper is to propose a new architecture for CgLp element which has a continuous output as opposed to conventional reset elements. This element will  drastically improves the transient response of the systems without jeopardizing the steady-state performance of the system by increasing higher-order harmonics. This paper shows that this architecture even reduces the higher-order harmonics by smoothing the reset jumps. Reset control systems are also known for having big jumps and peaks in their control input which can be a limiting factor in practical applications due to saturation. The proposed architecture will also improve this drawback. A guideline for tuning the propose architecture will also be provided.\\
The remainder of this paper is organized as follows: The first section will present the preliminaries of the study. The following section will present the continuous reset architecture. Section~\ref{sec:steady-state} will study the open-loop steady-state properties of the proposed architecture. The following two sections will numerically study the closed-loop transient and steady-state characteristics of the proposed controller. Section~\ref{sec:practice}, will verify the results by presenting the results of an experiment on a precision positioning system system and at last the paper concludes along with some tips for ongoing works. 

\section{Preliminaries}
This section will discuss the preliminaries of this study.
\subsection{Dynamics of Precision Motion Systems}
\label{sec:dynamics}
The first stage in precise control of a mechatronic system is to determine the dynamics of motion. A friction-less moving mass is the most basic mechatronic system. Its motion dynamics are represented by a double integrator. A DC motor or a voice-coil actuator are examples of such systems. In practice, the masses are usually constrained by springs and there is always some amounts of damping present, which creates a mass-spring-damper dynamics. Such dynamics in frequency domain has a constant spring line and a resonance peak in addition to the negative-sloped mass line.\\
Most of the precision motion setups are well-designed systems which can be modeled as mass-spring-damper systems or a cascade of them~\mbox{\cite{boeren2015iterative,schmidt2020design,8814888,rijlaarsdam2010frequency}}. Whether they are collocated or non-collocated systems, in practice, the cross-over frequency to control them is usually placed along the -2 slope mass line. Furthermore, the presence of integrator at lower frequencies, makes the overall open-loop frequency domain characteristics of precision motion systems to closely resemble a mass system.\\
This paper consists of an analytical analysis on steady-state properties of such systems and a numerical analysis of transient properties. Although the analytical steady-state analysis will be carried out for general motion plants, for the transient numerical analysis for the sake of generality and simplicity, a mass plant will be assumed. However, it will be shown in experimental results that the study hold for a mass-spring-damper system with higher frequency modes. 
\subsection{General Reset Controller}
The general form of reset controllers used in this study is as following:
\begin{align}
	\label{eq:reset}
	{{\sum }_{R}}:=\left\{ \begin{aligned}
		& {{{\dot{x}}}_{r}}(t)={{A_r}}{{x}_{r}}(t)+{{B_r}}e(t),&\text{if }e(t)\ne 0\\ 
		& {{x}_{r}}({{t}^{+}})={{A}_{\rho }}{{x}_{r}}(t),&\text{if }e(t)=0 \\ 
		& u(t)={{C_r}}{{x}_{r}}(t)+{{D_r}}e(t) \\ 
	\end{aligned} \right.
\end{align}
where $A_r,B_r,C_r,D_r$ denote the state space matrices of the Base Linear System (BLS) and reset matrix is denoted by $A_\rho=\text{diag}(\gamma_1,...,\gamma_n)$ which contains the reset coefficients for each state. $e(t)$ and $ u(t) $ represent the input and output for the reset controller, respectively.\\
A special type of reset elements which is of concern in this paper is First Order Reset Element (FORE). In the literature, this element is typically shown as $\cancelto{\gamma}{\frac{1}{{s}/{{{\omega }_{r }}+1}\;}}$, where $\omega_r$ is the corner frequency and the arrow indicates the resetting action and since element has only one resetting state, $A_\rho=\gamma$.
%In the spacial case of $A_\rho=I$, no reset will happen and the result system is called ``base linear system''.
\subsection{$H_\beta$ condition}
Among different criteria for stability of reset control systems~\mbox{\cite{banos2011reset,Guo:2015,polenkova2012stability,nevsic2008stability,vettori2014geometric,prieur2018analysis}}, despite of its conservativity, $H_\beta$ condition has gained attention because of simplicity and frequency domain applicability~\mbox{\cite{beker2004fundamental}}. In~\mbox{~\cite{dastjerdi2021newstability}}, the $H_\beta$ condition has been reformulated such that the frequency response functions of the controllers and the plant can be used directly. This method especially includes the case where the reset element is not the first element in the loop.\\ 
\begin{thm}
	\label{thm:h_beta}
	Let us denote frequency response functions of the open-loop BLS and the reset element as $O(j\omega)$ and $C_R
	(j\omega)$, respectively. And let the vector $\overrightarrow{\mathcal{N} }(\omega)\in {\mathbb{R}^{2}}$ be defined as $\overrightarrow{\mathcal{N}}(\omega)=[\begin{array}{ll}
		\mathcal{N}_{X} & \mathcal{N}_{Y}	\end{array}]^T$ in which
	\begin{align}
		\begin{aligned}
			\mathcal{N}_{X}&=\mathfrak{R}\left(O(j \omega)  \kappa(j \omega)\right),\\
			\mathcal{N}_{Y}&=\mathfrak{R}\left( \kappa(j \omega)C_R(j \omega) \right),
		\end{aligned}			
	\end{align}
where $\kappa(j \omega)=1+O^*(j\omega)$, $O^*(j\omega)$ is the conjugate of $O(j\omega)$ and $\mathfrak{R}(.)$ stands for the real part of a complex number. Let
\begin{equation}
\theta_{1}=\min _{\omega \in \mathbb{R}^{+}} \angle \overrightarrow{\mathcal{N}}(\omega) \text { and } \theta_{2}=\max _{\omega \in \mathbb{R}^{+}} \angle \overrightarrow{\mathcal{N}}(\omega).
\end{equation}
Then the $h_\beta$ condition for a reset control system is satisfied and its response is Uniformly Bounded-Input Bounded-State (UBIBS) stable if
\begin{equation}
	\left(-\frac{\pi}{2}<\theta_{1}<\pi\right) \wedge\left(-\frac{\pi}{2}<\theta_{2}<\pi\right) \wedge\left(\theta_{2}-\theta_{1}<\pi\right).
\end{equation}
\end{thm}
\subsection{Describing Functions}
Describing function analysis is the known approach in literature for approximation of frequency response of nonlinear systems like reset controllers~\cite{guo2009frequency}. However, the DF method only takes the first harmonic of Fourier series decomposition of the output into account and neglects the effects of the higher order harmonics. This simplification can be significantly inaccurate under certain circumstances~\cite{karbasizadeh2020benefiting}. The ``Higher Order Sinusoidal Input Describing Function'' (HOSIDF) method has been introduced in~\cite{nuij2006higher} to provide more accurate  information about the frequency response of nonlinear systems by investigation of higher-order harmonics of the Fourier series decomposition. In other words, in this method, the nonlinear element will be replaced by a virtual harmonic generator. This method was developed in~\cite{saikumar2021loop} for reset elements defined by Eq.~({\ref{eq:reset}}) as follows:
\begin{align}  \nonumber
	\label{eq:hosidf}
	& H_n(\omega)=\left\{ \begin{aligned}
		& C_r{{(j\omega I-A_r)}^{-1}}(I+j{{\Theta }}(\omega ))B_r+D_r,\quad n=1\\ 
		& C_r{{(j\omega nI-A_r)}^{-1}}j{{\Theta }}(\omega )B_r,\quad~~\qquad\text{odd }n> 2\\ 
		& 0,\qquad\qquad\qquad\qquad\qquad\qquad\qquad~\text{even }n\ge 2\\ 
	\end{aligned} \right. \\
	&\begin{aligned}
		& {{\Theta }}(\omega )=-\frac{2{{\omega }^{2}}}{\pi }\Delta (\omega )[{{\Gamma }}(\omega )-{{\Lambda }^{-1}}(\omega )] \\  
		& \Lambda (\omega )={{\omega }^{2}}I+{{A_r}^{2}} \\  
		& \Delta (\omega )=I+{{e}^{\frac{\pi }{\omega }A_r}} \\  
		& {{\Delta }_{\rho}}(\omega )=I+{{A}_{\rho}}{{e}^{\frac{\pi }{\omega }A_r}} \\  
		& {{\Gamma }}(\omega )={\Delta }_{\rho}^{-1}(\omega ){{A}_{\rho}}\Delta (\omega ){{\Lambda }^{-1}}(\omega ) \\
	\end{aligned} 
\end{align}
where $H_n(\omega)$ is the $n^{\text{th}}$ harmonic describing function for sinusoidal input with frequency of $\omega$. \\
\begin{figure}[t!]
	\centering
	\includegraphics[width=\textwidth]{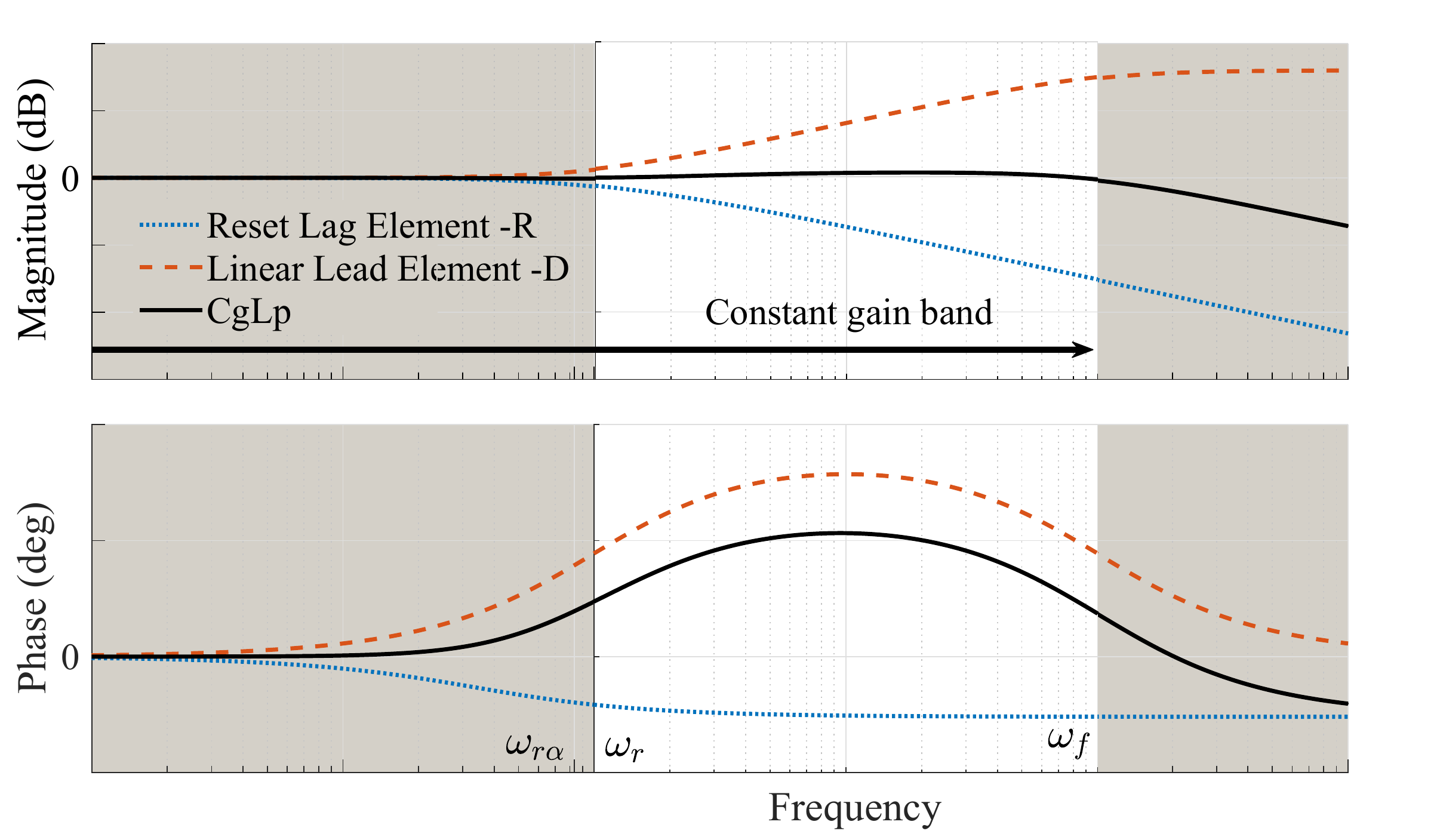}
	\caption{The concept of using combination of a reset lag and a linear lead element to form a CgLp element. The figure is  from~\cite{saikumar2019constant}.}
	\label{fig:cglp}
\end{figure}
\subsection{CgLp}
CgLp is a broadband phase compensation reset element which has a first harmonic constant gain behaviour while providing a phase lead~\cite{saikumar2019constant}. This element consists in a reset lag element in series with a linear lead filter, namely ${\sum}_R$ and $D$, respectively. For FORE CgLp:
\begin{align}
	\label{eq:fore}
	&{\sum}_R=\cancelto{\gamma}{\frac{1}{{s}/{{{\omega }_{r }}+1}\;}},&D(s)=\frac{{s}/{{{\omega }_{r\alpha}}}\;+1}{{s}/{{{\omega }_{f}}}\;+1}
\end{align}
where $\omega_{r\alpha}=\alpha \omega_r$, $\alpha$ is a tuning parameter accounting for a shift in corner frequency of the filter due to resetting action,  and $[\omega_{r},\omega_{f}]$ is the frequency range where the CgLp will provide the required phase lead. The arrow indicates the resetting action as described in Eq.~\eqref{eq:reset}.\\
CgLp provides the phase lead by using the reduced phase lag of reset lag element in combination with a corresponding lead element to create broadband phase lead. Ideally, the gain of the reset lag element should be canceled out by the gain of the corresponding linear lead element, which creates a constant gain behavior. The concept is depicted in Fig.~\ref{fig:cglp}.\\
\section{Proposed Architecture for Continuous Reset (CR) Elements}
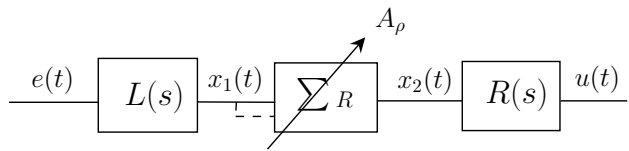
\begin{figure}[t!]
	\centering
	\resizebox{0.7\textwidth}{!}{

		\tikzset{every picture/.style={line width=0.75pt}} %set default line width to 0.75pt        
		
		\begin{tikzpicture}[x=0.75pt,y=0.75pt,yscale=-1,xscale=1]
			%uncomment if require: \path (0,111); %set diagram left start at 0, and has height of 111
			
			%Shape: Rectangle [id:dp18437809240624103] 
			\draw  [line width=0.75]  (63,40) -- (133,40) -- (133,92.07) -- (63,92.07) -- cycle ;
			%Shape: Rectangle [id:dp9694418703446797] 
			\draw  [line width=0.75]  (187.5,43.17) -- (260,43.17) -- (260,92) -- (187.5,92) -- cycle ;
			%Straight Lines [id:da10782120214651458] 
			\draw [line width=0.75]    (182.5,104.08) -- (249.53,27.34) ;
			\draw [shift={(251.5,25.08)}, rotate = 491.13] [fill={rgb, 255:red, 0; green, 0; blue, 0 }  ][line width=0.08]  [draw opacity=0] (10.72,-5.15) -- (0,0) -- (10.72,5.15) -- (7.12,0) -- cycle    ;
			%Straight Lines [id:da5144772536260698] 
			\draw [line width=0.75]    (0,71) -- (63,71) ;
			%Straight Lines [id:da1326344222324929] 
			\draw [line width=0.75]    (133,70.75) -- (188,71) ;
			%Shape: Rectangle [id:dp3301486819736583] 
			\draw  [line width=0.75]  (320,42.17) -- (389,42.17) -- (389,91) -- (320,91) -- cycle ;
			%Straight Lines [id:da6269618740664133] 
			\draw [line width=0.75]    (260,70) -- (320,70) ;
			%Straight Lines [id:da7313027750694783] 
			\draw [line width=0.75]    (389.5,70) -- (436.5,70) ;
			%Straight Lines [id:da06762622941050234] 
			\draw  [dash pattern={on 4.5pt off 4.5pt}]  (160.5,70.87) -- (160,81) ;
			%Straight Lines [id:da7759977624303527] 
			\draw  [dash pattern={on 4.5pt off 4.5pt}]  (160,81) -- (187,81) ;
			
			% Text Node
			\draw (415.5,54.5) node  [font=\Large]  {$u( t)$};
			% Text Node
			\draw (31,55.5) node  [font=\Large]  {$e( t)$};
			% Text Node
			\draw (257,3.4) node [anchor=north west][inner sep=0.75pt]  [font=\Large]  {$A_{\rho }$};
			% Text Node
			\draw (80,54.15) node [anchor=north west][inner sep=0.75pt]  [font=\LARGE]  {$L( s)$};
			% Text Node
			\draw (201,52.4) node [anchor=north west][inner sep=0.75pt]  [font=\LARGE]  {$\sum{}_{R}$};
			% Text Node
			\draw (335,52.4) node [anchor=north west][inner sep=0.75pt]  [font=\LARGE]  {$R( s)$};
			% Text Node
			\draw (160,55.5) node  [font=\Large]  {$x_{1}( t)$};
			% Text Node
			\draw (294,56) node  [font=\Large]  {$x_{2}( t)$};

		\end{tikzpicture}
	}
	\caption{Proposed architecture for reset elements which includes a lead element, $L(s)$ before the reset element and ints inverse after the reset element. The proposed lead is $L(s)=\frac{s/\omega_l+1}{s/\omega_h+1}$ and $R(s)=\frac{1}{s/\omega_l+1}$.}
	\label{fig:new_arch}
\end{figure}
The new architecture which this paper proposes consists of adding a first-order lag element, $R(s)$, after the reset element and adding the inverse of it, which is basically a lead element, after the reset element. Fig.~\ref{fig:new_arch} depicts the new architecture in which
\begin{equation}
	L(s)=\frac{s/\omega_l+1}{s/\omega_h+1},\quad R(s)=\frac{1}{s/\omega_l+1}.
	\label{eq:L}
\end{equation}
In the ideal case, $L(s)=R^{-1}(s)$, however, in order to make $L(s)$ proper and realizable, the presence of the denominator in $L(s)$ is necessary. Nevertheless, assuming $\omega_h$ is large enough, $R(s)\approx L^{-1}(s)$ in low frequencies. In the context of linear control systems, adding these two elements would almost have no effect on the output of the system in lower frequencies and improve the noise attenuation behaviour at higher frequencies, provided the internal states stability. However, in the context of nonlinear control systems, the output of the system will be changed significantly. \\
In this new architecture the resetting condition is changed from $e(t)=0$ to $x_1(t)=0$. Again considering that $\omega_h$ is large enough, the new resetting condition can be approximated as  
\begin{equation}\label{eq:reset_law}
	x_1(t)={\dot{e}(t)}/\omega_l+e(t)=0.
\end{equation}
The new reset element resets based on a linear combination of $e(t)$ and $\dot{e}(t)$, where $\omega_l$ determines the weight of each. In closed loop, $e(t)$ and $\dot{e}(t)$ are the error and its differentiation.\\
\begin{rem}
	\label{rem:stability}
	According to Theorem~\mbox{\ref{thm:h_beta}}, a reset element in CR architecture has the same stability properties as standing alone, as long as $O(s)$ stays the same, i.e., $R(s)$ and $L(s)$ cancel each other in linear domain. In other words, adding $L(s)$ and $R(s)$ in CR architecture, does not affect the stability properties of the reset control system. 
	However, for the architecture presented in this paper, the additional condition is $\omega_h \gg \omega_{r}$ and  $\omega_h \gg \omega_{c}$, where $\omega_c$ is the cross-over frequency.
\end{rem}
\begin{thm}
	\label{thm:cont}
	The output of the proposed architecture is continuous as opposed to  ${\sum{}}_R$ alone.
\end{thm}
\begin{proof}
	If the reset instants are $\{t_k \mid k=1,2,3,\cdots\}$, from Eq.~\eqref{eq:reset} and Fig.~\ref{fig:new_arch}, it can be seen that 
	\begin{align}
		\label{eq:reset_x1_x2}
		{{\sum }_{R}}:=\left\{ \begin{aligned}
			& {{{\dot{x}}}_{r}}(t)={{A_r}}{{x}_{r}}(t)+{{B_r}}x_1(t),&\text{if }t\ne t_k\\ 
			& {{x}_{r}}({{t}^{+}})={{A}_{\rho }}{{x}_{r}}(t),&\text{if }t=t_k \\ 
			& x_2(t)={{C_r}}{{x}_{r}}(t)+{{D_r}}x_1(t) \\ 
		\end{aligned} \right.
	\end{align}
	It is readily obvious that $x_2(t)$ is continuous on $(t_{k-1},t_k)$ and $(t_k,t_{k+1})$. However,
	\begin{equation}
		\underset{t\to t_{k}^{-}}{\mathop{\lim }}\,{{x}_{2}}(t)\ne \underset{t\to t_{k}^{+}}{\mathop{\lim }}\,{{x}_{2}}(t)
	\end{equation}
	and thus it is discontinuous. Nevertheless, for $u(t)$ one can write
	\begin{align}
		\nonumber
		&\underset{{{t}_{k-1}}<t\leq{{t}_{k}}}{\mathop{u(t)}}\,=\\
		&{{\omega }_{l}}\left( {{e}^{-{{\omega }_{l}}(t-{{t}_{k-1}})}}u({{t}_{k-1}})+\int_{{{t}_{k-1}}}^{t}{{{e}^{-{{\omega }_{l}}(t-\tau )}}{{x}_{2}}(\tau )d\tau } \right).
	\end{align}
	It can be readily seen that
	\begin{align}
		\nonumber
		&\underset{t\to t_{k}^{-}}{\mathop{\lim }}\,{{u}}(t)= \underset{t\to t_{k}^{+}}{\mathop{\lim }}\,{{u}}(t)=\\
		&{{\omega }_{d}}\left( {{e}^{-{{\omega }_{l}}(t_k-{{t}_{k-1}})}}u({{t}_{k-1}})+\int_{{{t}_{k-1}}}^{t_k}{{{e}^{-{{\omega }_{l}}(t_k-\tau )}}{{x}_{2}}(\tau )d\tau } \right).
	\end{align}
\end{proof}
In addition to making the reset element output continuous, other motivations to use this architecture can be described in terms of steady-state and transient response of system, which will be discussed in details in following sections. 
\section{Open-Loop Steady-State Properties of the CR Architecture}
\label{sec:steady-state}
Frequency domain analysis is the popular approach for study of the steady-state response of a system. However, as mentioned earlier, because of the nonlinearity of reset elements, that is not directly possible. The DF and HOSIDF methods are two approaches to approximate a frequency response for a reset  control systems, where DF can be regarded as a special case of HOSIDF in which, only the first-order harmonic is studied. In order to illustrate how the HOSIDF approach can be used for the CR  architecture proposed, one can refer to Fig.~\ref{fig:HOSIDF}. \\
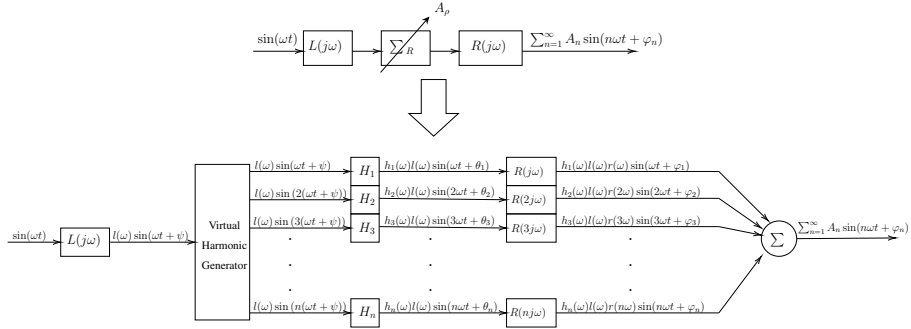
\begin{figure}[t!]
	\centering
	\resizebox{\textwidth}{!}{

		\tikzset{every picture/.style={line width=0.75pt}} %set default line width to 0.75pt        
		
		\begin{tikzpicture}[x=0.75pt,y=0.75pt,yscale=-1,xscale=1]
			%uncomment if require: \path (0,553); %set diagram left start at 0, and has height of 553
			
			%Straight Lines [id:da4148683692834523] 
			\draw    (352,80) -- (420,80) ;
			\draw [shift={(422,80)}, rotate = 180] [color={rgb, 255:red, 0; green, 0; blue, 0 }  ][line width=0.75]    (10.93,-3.29) .. controls (6.95,-1.4) and (3.31,-0.3) .. (0,0) .. controls (3.31,0.3) and (6.95,1.4) .. (10.93,3.29)   ;
			%Straight Lines [id:da8117644498811956] 
			\draw    (492,80) -- (530,80) ;
			\draw [shift={(532,80)}, rotate = 180] [color={rgb, 255:red, 0; green, 0; blue, 0 }  ][line width=0.75]    (10.93,-3.29) .. controls (6.95,-1.4) and (3.31,-0.3) .. (0,0) .. controls (3.31,0.3) and (6.95,1.4) .. (10.93,3.29)   ;
			%Shape: Rectangle [id:dp5381773113987138] 
			\draw  [line width=0.75]  (533,51.17) -- (602,51.17) -- (602,100) -- (533,100) -- cycle ;
			%Straight Lines [id:da856152467087228] 
			\draw [line width=0.75]    (532,110) -- (599.03,33.26) ;
			\draw [shift={(601,31)}, rotate = 491.13] [fill={rgb, 255:red, 0; green, 0; blue, 0 }  ][line width=0.08]  [draw opacity=0] (10.72,-5.15) -- (0,0) -- (10.72,5.15) -- (7.12,0) -- cycle    ;
			%Shape: Rectangle [id:dp9255909288119744] 
			\draw  [line width=0.75]  (423,51.17) -- (492,51.17) -- (492,100) -- (423,100) -- cycle ;
			%Shape: Rectangle [id:dp08637720914261071] 
			\draw  [line width=0.75]  (643,51.17) -- (732,51.17) -- (732,100) -- (643,100) -- cycle ;
			%Straight Lines [id:da030916293504181347] 
			\draw    (602,80) -- (640,80) ;
			\draw [shift={(642,80)}, rotate = 180] [color={rgb, 255:red, 0; green, 0; blue, 0 }  ][line width=0.75]    (10.93,-3.29) .. controls (6.95,-1.4) and (3.31,-0.3) .. (0,0) .. controls (3.31,0.3) and (6.95,1.4) .. (10.93,3.29)   ;
			%Straight Lines [id:da4455364210913155] 
			\draw    (732,80) -- (890,80) ;
			\draw [shift={(892,80)}, rotate = 180] [color={rgb, 255:red, 0; green, 0; blue, 0 }  ][line width=0.75]    (10.93,-3.29) .. controls (6.95,-1.4) and (3.31,-0.3) .. (0,0) .. controls (3.31,0.3) and (6.95,1.4) .. (10.93,3.29)   ;
			%Straight Lines [id:da7727487692859847] 
			\draw    (5,350) -- (78,350) ;
			\draw [shift={(80,350)}, rotate = 180] [color={rgb, 255:red, 0; green, 0; blue, 0 }  ][line width=0.75]    (10.93,-3.29) .. controls (6.95,-1.4) and (3.31,-0.3) .. (0,0) .. controls (3.31,0.3) and (6.95,1.4) .. (10.93,3.29)   ;
			%Shape: Rectangle [id:dp1088627732919556] 
			\draw   (270,240) -- (347,240) -- (347,460) -- (270,460) -- cycle ;
			%Straight Lines [id:da774065695443207] 
			\draw    (347,250) -- (488,250) ;
			\draw [shift={(490,250)}, rotate = 180] [color={rgb, 255:red, 0; green, 0; blue, 0 }  ][line width=0.75]    (10.93,-3.29) .. controls (6.95,-1.4) and (3.31,-0.3) .. (0,0) .. controls (3.31,0.3) and (6.95,1.4) .. (10.93,3.29)   ;
			%Straight Lines [id:da3318613910527355] 
			\draw    (347,290) -- (444.6,290) -- (488,290) ;
			\draw [shift={(490,290)}, rotate = 180] [color={rgb, 255:red, 0; green, 0; blue, 0 }  ][line width=0.75]    (10.93,-3.29) .. controls (6.95,-1.4) and (3.31,-0.3) .. (0,0) .. controls (3.31,0.3) and (6.95,1.4) .. (10.93,3.29)   ;
			%Straight Lines [id:da4614361797688036] 
			\draw    (347,330) -- (488,330) ;
			\draw [shift={(490,330)}, rotate = 180] [color={rgb, 255:red, 0; green, 0; blue, 0 }  ][line width=0.75]    (10.93,-3.29) .. controls (6.95,-1.4) and (3.31,-0.3) .. (0,0) .. controls (3.31,0.3) and (6.95,1.4) .. (10.93,3.29)   ;
			%Straight Lines [id:da02922800165702255] 
			\draw    (347,450) -- (488,450) ;
			\draw [shift={(490,450)}, rotate = 180] [color={rgb, 255:red, 0; green, 0; blue, 0 }  ][line width=0.75]    (10.93,-3.29) .. controls (6.95,-1.4) and (3.31,-0.3) .. (0,0) .. controls (3.31,0.3) and (6.95,1.4) .. (10.93,3.29)   ;
			%Shape: Rectangle [id:dp18565363100901222] 
			\draw   (490,230) -- (530,230) -- (530,270) -- (490,270) -- cycle ;
			%Shape: Rectangle [id:dp6614955637453066] 
			\draw   (490,270) -- (530,270) -- (530,310) -- (490,310) -- cycle ;
			%Shape: Rectangle [id:dp3271990072638937] 
			\draw   (490,310) -- (530,310) -- (530,350) -- (490,350) -- cycle ;
			%Shape: Rectangle [id:dp1589627159320326] 
			\draw   (490,430) -- (530,430) -- (530,470) -- (490,470) -- cycle ;
			%Straight Lines [id:da08891805506043982] 
			\draw    (530,250) -- (708,250) ;
			\draw [shift={(710,250)}, rotate = 180] [color={rgb, 255:red, 0; green, 0; blue, 0 }  ][line width=0.75]    (10.93,-3.29) .. controls (6.95,-1.4) and (3.31,-0.3) .. (0,0) .. controls (3.31,0.3) and (6.95,1.4) .. (10.93,3.29)   ;
			%Straight Lines [id:da1025605844063131] 
			\draw    (530,288) -- (708,289.98) ;
			\draw [shift={(710,290)}, rotate = 180.64] [color={rgb, 255:red, 0; green, 0; blue, 0 }  ][line width=0.75]    (10.93,-3.29) .. controls (6.95,-1.4) and (3.31,-0.3) .. (0,0) .. controls (3.31,0.3) and (6.95,1.4) .. (10.93,3.29)   ;
			%Straight Lines [id:da6239618630881274] 
			\draw    (530,328) -- (708,329.98) ;
			\draw [shift={(710,330)}, rotate = 180.64] [color={rgb, 255:red, 0; green, 0; blue, 0 }  ][line width=0.75]    (10.93,-3.29) .. controls (6.95,-1.4) and (3.31,-0.3) .. (0,0) .. controls (3.31,0.3) and (6.95,1.4) .. (10.93,3.29)   ;
			%Straight Lines [id:da8937276935254859] 
			\draw    (530,450) -- (708,450) ;
			\draw [shift={(710,450)}, rotate = 180] [color={rgb, 255:red, 0; green, 0; blue, 0 }  ][line width=0.75]    (10.93,-3.29) .. controls (6.95,-1.4) and (3.31,-0.3) .. (0,0) .. controls (3.31,0.3) and (6.95,1.4) .. (10.93,3.29)   ;
			%Shape: Rectangle [id:dp8143039059105412] 
			\draw  [line width=0.75]  (80,330) -- (149,330) -- (149,370) -- (80,370) -- cycle ;
			%Straight Lines [id:da8307075024933832] 
			\draw    (150,350) -- (268,350) ;
			\draw [shift={(270,350)}, rotate = 180] [color={rgb, 255:red, 0; green, 0; blue, 0 }  ][line width=0.75]    (10.93,-3.29) .. controls (6.95,-1.4) and (3.31,-0.3) .. (0,0) .. controls (3.31,0.3) and (6.95,1.4) .. (10.93,3.29)   ;
			%Shape: Rectangle [id:dp5845585680675112] 
			\draw   (710,230) -- (780,230) -- (780,270) -- (710,270) -- cycle ;
			%Shape: Rectangle [id:dp452675428244641] 
			\draw   (710,270) -- (780,270) -- (780,310) -- (710,310) -- cycle ;
			%Shape: Rectangle [id:dp43966085780043485] 
			\draw   (710,310) -- (780,310) -- (780,350) -- (710,350) -- cycle ;
			%Shape: Rectangle [id:dp3400406038344055] 
			\draw   (710,430) -- (780,430) -- (780,470) -- (710,470) -- cycle ;
			%Straight Lines [id:da5353361508399004] 
			\draw    (780,250) -- (1010,250) ;
			%Straight Lines [id:da569863886775481] 
			\draw    (780,288) -- (1010,288) ;
			%Straight Lines [id:da25704190504696456] 
			\draw    (780,328) -- (1010,328) ;
			%Straight Lines [id:da5197694242620867] 
			\draw    (780,450) -- (1010,450) ;
			%Shape: Circle [id:dp2668656554108857] 
			\draw   (1070,345) .. controls (1070,331.19) and (1081.19,320) .. (1095,320) .. controls (1108.81,320) and (1120,331.19) .. (1120,345) .. controls (1120,358.81) and (1108.81,370) .. (1095,370) .. controls (1081.19,370) and (1070,358.81) .. (1070,345) -- cycle ;
			%Straight Lines [id:da9781948254906447] 
			\draw    (1010,250) -- (1078.59,318.59) ;
			\draw [shift={(1080,320)}, rotate = 225] [color={rgb, 255:red, 0; green, 0; blue, 0 }  ][line width=0.75]    (10.93,-3.29) .. controls (6.95,-1.4) and (3.31,-0.3) .. (0,0) .. controls (3.31,0.3) and (6.95,1.4) .. (10.93,3.29)   ;
			%Straight Lines [id:da7987299975885838] 
			\draw    (1010,288) -- (1068.36,328.85) ;
			\draw [shift={(1070,330)}, rotate = 214.99] [color={rgb, 255:red, 0; green, 0; blue, 0 }  ][line width=0.75]    (10.93,-3.29) .. controls (6.95,-1.4) and (3.31,-0.3) .. (0,0) .. controls (3.31,0.3) and (6.95,1.4) .. (10.93,3.29)   ;
			%Straight Lines [id:da9358405934556593] 
			\draw    (1010,328) -- (1068.04,339.61) ;
			\draw [shift={(1070,340)}, rotate = 191.31] [color={rgb, 255:red, 0; green, 0; blue, 0 }  ][line width=0.75]    (10.93,-3.29) .. controls (6.95,-1.4) and (3.31,-0.3) .. (0,0) .. controls (3.31,0.3) and (6.95,1.4) .. (10.93,3.29)   ;
			%Straight Lines [id:da9654261652576999] 
			\draw    (1010,450) -- (1068.8,371.6) ;
			\draw [shift={(1070,370)}, rotate = 486.87] [color={rgb, 255:red, 0; green, 0; blue, 0 }  ][line width=0.75]    (10.93,-3.29) .. controls (6.95,-1.4) and (3.31,-0.3) .. (0,0) .. controls (3.31,0.3) and (6.95,1.4) .. (10.93,3.29)   ;
			%Straight Lines [id:da21122222102029098] 
			\draw    (1120,345) -- (1258,344.01) ;
			\draw [shift={(1260,344)}, rotate = 539.5899999999999] [color={rgb, 255:red, 0; green, 0; blue, 0 }  ][line width=0.75]    (10.93,-3.29) .. controls (6.95,-1.4) and (3.31,-0.3) .. (0,0) .. controls (3.31,0.3) and (6.95,1.4) .. (10.93,3.29)   ;
			%Down Arrow [id:dp8967250787765764] 
			\draw   (572,168) -- (589.5,168) -- (589.5,120) -- (624.5,120) -- (624.5,168) -- (642,168) -- (607,200) -- cycle ;
			
			% Text Node
			\draw (357,52.4) node [anchor=north west][inner sep=0.75pt]   [font=\Large] {$\sin( \omega t)$};
			% Text Node
			\draw (428,62.4) node [anchor=north west][inner sep=0.75pt]  [font=\Large]  {$L( j\omega )$};
			% Text Node
			\draw (606.5,9.32) node [anchor=north west][inner sep=0.75pt]  [font=\Large]  {$A_{\rho }$};
			% Text Node
			\draw (543,62.4) node [anchor=north west][inner sep=0.75pt]  [font=\Large]  {${\sum{}}_{R}$};
			% Text Node
			\draw (658,62.4) node [anchor=north west][inner sep=0.75pt]  [font=\Large]  {$R( j\omega )$};
			% Text Node
			\draw (742,52.4) node [anchor=north west][inner sep=0.75pt]  [font=\Large]  {$\sum _{n=1}^{\infty } A_{n}\sin( n\omega t+\varphi _{n})$};
			% Text Node
			\draw (11,332.4) node [anchor=north west][inner sep=0.75pt]   [font=\large] {$\sin( \omega t)$};
			% Text Node
			\draw (278,319) node [anchor=north west][inner sep=0.75pt]  [font=\large] [align=center] {{\fontfamily{ptm}\selectfont Virtual}\\{\fontfamily{ptm}\selectfont Harmonic }\\{\fontfamily{ptm}\selectfont Generator}};
			% Text Node
			\draw (353,232.4) node [anchor=north west][inner sep=0.75pt]  [font=\large]  {$l( \omega )\sin( \omega t+\psi )$};
			% Text Node
			\draw (353,272.4) node [anchor=north west][inner sep=0.75pt]  [font=\large]  {$l( \omega )\sin\left( 2( \omega t+\psi )\right)$};
			% Text Node
			\draw (400,343) node [anchor=north west][inner sep=0.75pt]  [font=\Large] [align=left] {\textbf{.}\\\textbf{.}\\\textbf{.}};
			% Text Node
			\draw (500,240.4) node [anchor=north west][inner sep=0.75pt]  [font=\Large]  {$H_{1}$};
			% Text Node
			\draw (500,280.4) node [anchor=north west][inner sep=0.75pt]   [font=\Large]   {$H_{2}$};
			% Text Node
			\draw (500,320.4) node [anchor=north west][inner sep=0.75pt]   [font=\Large]   {$H_{3}$};
			% Text Node
			\draw (500,440.4) node [anchor=north west][inner sep=0.75pt]   [font=\Large]   {$H_{n}$};
			% Text Node
			\draw (536,232.4) node [anchor=north west][inner sep=0.75pt]   [font=\large]  {$h_{1}( \omega ) l( \omega )\sin( \omega t+\theta _{1})$};
			% Text Node
			\draw (536,270.4) node [anchor=north west][inner sep=0.75pt]   [font=\large]  {$h_{2}( \omega ) l( \omega )\sin( 2\omega t+\theta _{2})$};
			% Text Node
			\draw (536,310.4) node [anchor=north west][inner sep=0.75pt]  [font=\large]   {$h_{3}( \omega ) l( \omega )\sin( 3\omega t+\theta _{3})$};
			% Text Node
			\draw (536,432.4) node [anchor=north west][inner sep=0.75pt]   [font=\large]  {$h_{n}( \omega ) l( \omega )\sin( n\omega t+\theta _{n})$};
			% Text Node
			\draw (87,336.4) node [anchor=north west][inner sep=0.75pt]  [font=\Large]  {$L( j\omega )$};
			% Text Node
			\draw (152,332.4) node [anchor=north west][inner sep=0.75pt]   [font=\large]  {$l( \omega )\sin( \omega t+\psi )$};
			% Text Node
			\draw (598,343) node [anchor=north west][inner sep=0.75pt]  [font=\Large] [align=left] {\textbf{.}\\\textbf{.}\\\textbf{.}};
			% Text Node
			\draw (719,241.4) node [anchor=north west][inner sep=0.75pt]   [font=\large]  {$R( j\omega )$};
			% Text Node
			\draw (715,281.4) node [anchor=north west][inner sep=0.75pt]   [font=\large]  {$R( 2j\omega )$};
			% Text Node
			\draw (715,321.4) node [anchor=north west][inner sep=0.75pt]   [font=\large]  {$R( 3j\omega )$};
			% Text Node
			\draw (713,441.4) node [anchor=north west][inner sep=0.75pt]   [font=\large]  {$R( nj\omega )$};
			% Text Node
			\draw (786,232.4) node [anchor=north west][inner sep=0.75pt]    [font=\large] {$ { h_{1}( \omega ) l( \omega )} { r( \omega )} \sin( \omega t+\varphi _{1})$};
			% Text Node
			\draw (786,270.4) node [anchor=north west][inner sep=0.75pt]   [font=\large]  {$  {h_{2}( \omega ) l( \omega ) } {r( 2\omega )} \sin( 2\omega t+\varphi _{2})$};
			% Text Node
			\draw (786,310.4) node [anchor=north west][inner sep=0.75pt]   [font=\large]  {$ { h_{3}( \omega ) l( \omega )} { r( 3\omega )} \sin( 3\omega t+\varphi _{3})$};
			% Text Node
			\draw (786,432.4) node [anchor=north west][inner sep=0.75pt]   [font=\large]  {$  {h_{n}( \omega ) l( \omega )}{ r( n\omega )} \sin( n\omega t+\varphi _{n})$};
			% Text Node
			\draw (1081,335) node [anchor=north west][inner sep=0.75pt] [font=\Large]   {$\sum $};
			% Text Node
			\draw (881,343) node [anchor=north west][inner sep=0.75pt]  [font=\Large] [align=left] {\textbf{.}\\\textbf{.}\\\textbf{.}};
			% Text Node
			\draw (1120,320) node [anchor=north west][inner sep=0.75pt]   [font=\large]  {$\sum _{n=1}^{\infty } A_{n}\sin( n\omega t+\varphi _{n})$};
			% Text Node
			\draw (351,312.4) node [anchor=north west][inner sep=0.75pt]   [font=\large]  {$l( \omega )\sin\left( 3( \omega t+\psi )\right)$};
			% Text Node
			\draw (351,432.4) node [anchor=north west][inner sep=0.75pt]  [font=\large]   {$l( \omega )\sin\left( n( \omega t+\psi )\right)$};

		\end{tikzpicture}
	}
	\caption{Representation of HOSIDF for open-loop analysis of the new architecture proposed. $l(\omega)=\lvert L(j\omega) \rvert$, $r(\omega)=\lvert R(j\omega) \rvert$ and $h_n(\omega)=\lvert H_n(j\omega) \rvert$, where $H_n(\omega)$ can be obtained from Eq.~\eqref{eq:hosidf} for ${\sum{}}_R$.}
	\label{fig:HOSIDF}
\end{figure}
\begin{figure}[t!]
	\centering
	\begin{subfigure}{0.9\textwidth}
	\includegraphics[width=0.9\textwidth]{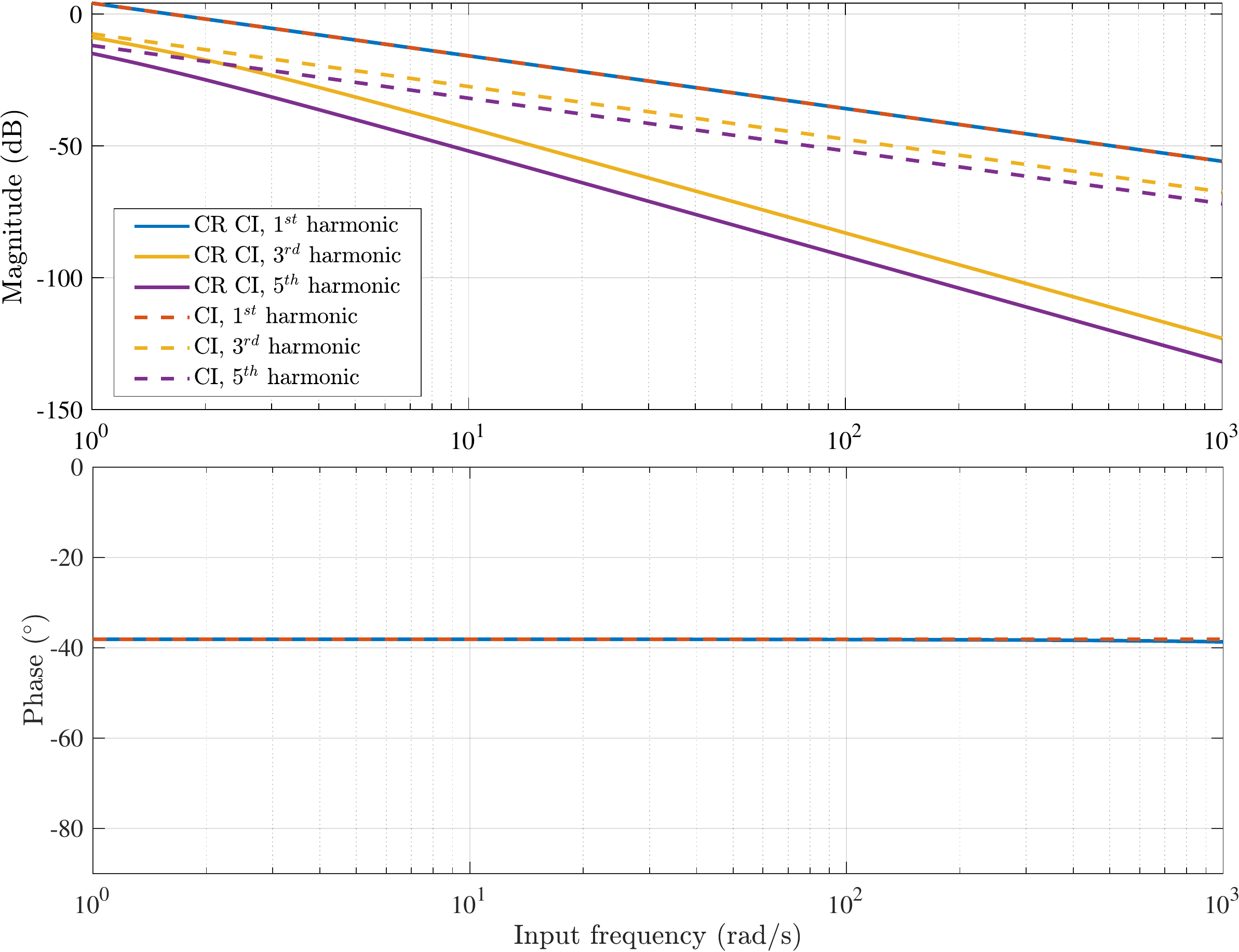}%
		\caption{HOSIDF of CI compared to CR CI. $\omega_l=10$, $\omega_h=1e4$ and $\gamma=0$. }
		\label{fig:clegg_hosidf}
	\end{subfigure}\\
\begin{subfigure}{0.9\textwidth}
	\includegraphics[width=0.9\textwidth]{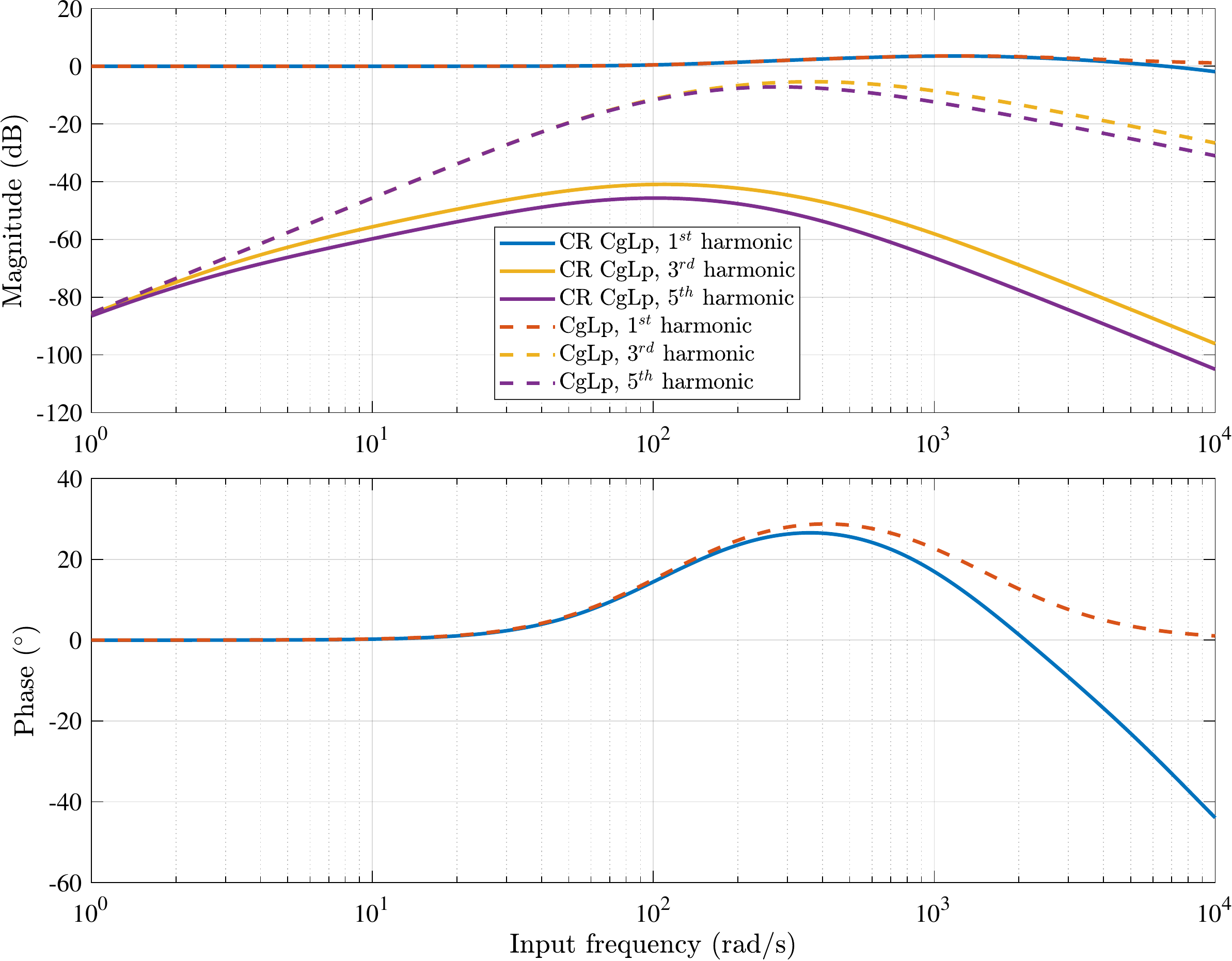}%
		\caption{HOSIDF of CgLp compared to CR CgLp. $\omega_l=10$, $\omega_h=1e4$, $\omega_r=100$, $\omega_f=1500$ and $\gamma=0.11$.}
		\label{fig:cglp_hosidf}
	\end{subfigure}
	\caption{HOSIDF of CI and CgLp compared to their CR architecture proposed in this paper.}
	\label{fig:hosidf_clegg_cplg}
\end{figure}   
\begin{figure}[!t]
	\centering
	\begin{subfigure}{\textwidth}
	\includegraphics[width=\textwidth]{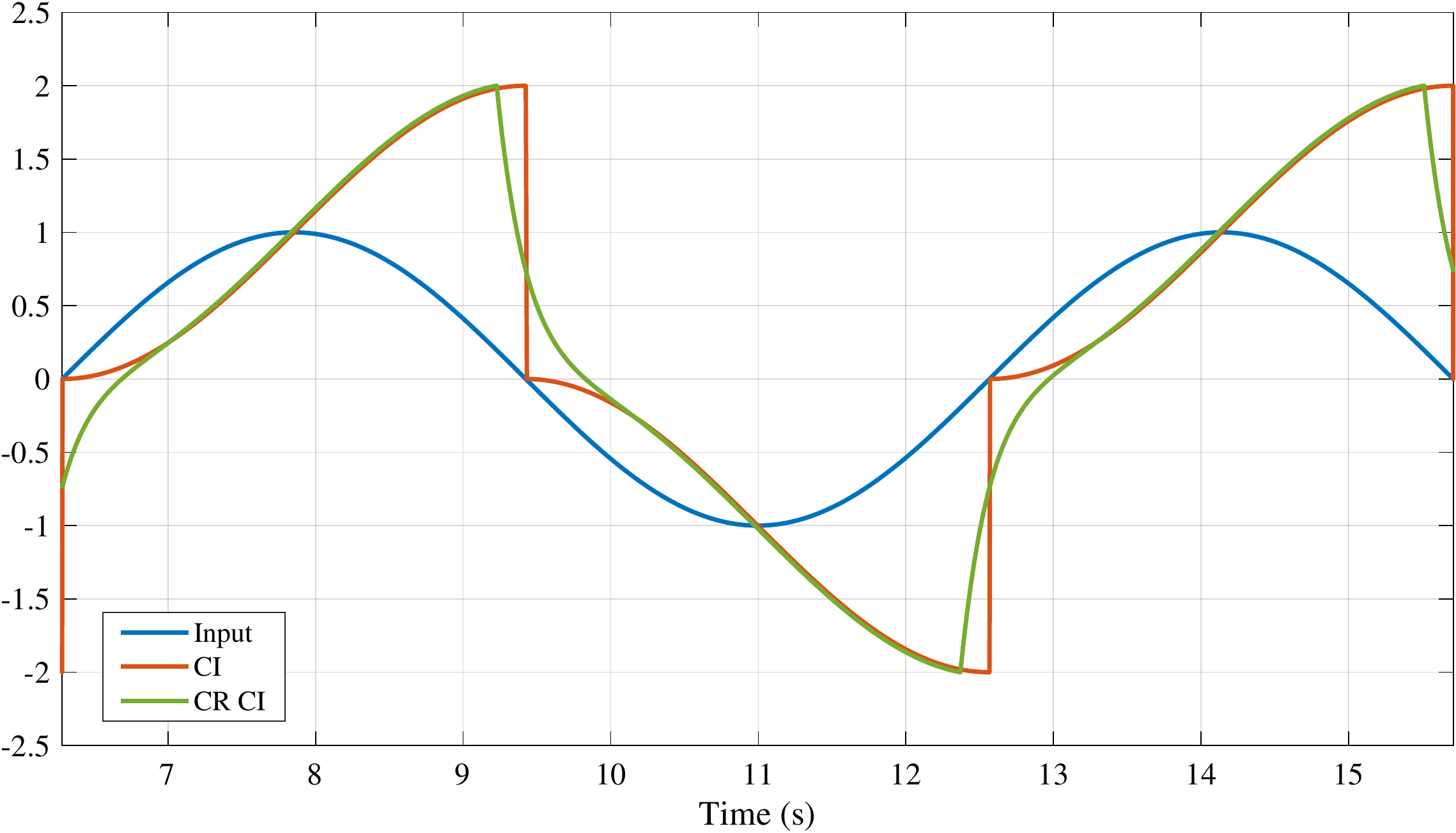}%
	\caption{Sinusoidal response of Clegg Integrator (CI) compared to CR CI. Input is $\sin(t)$.}
		\label{fig:clegg_sin}
	\end{subfigure}\\
	\begin{subfigure}{\textwidth}
		\includegraphics[width=\textwidth]{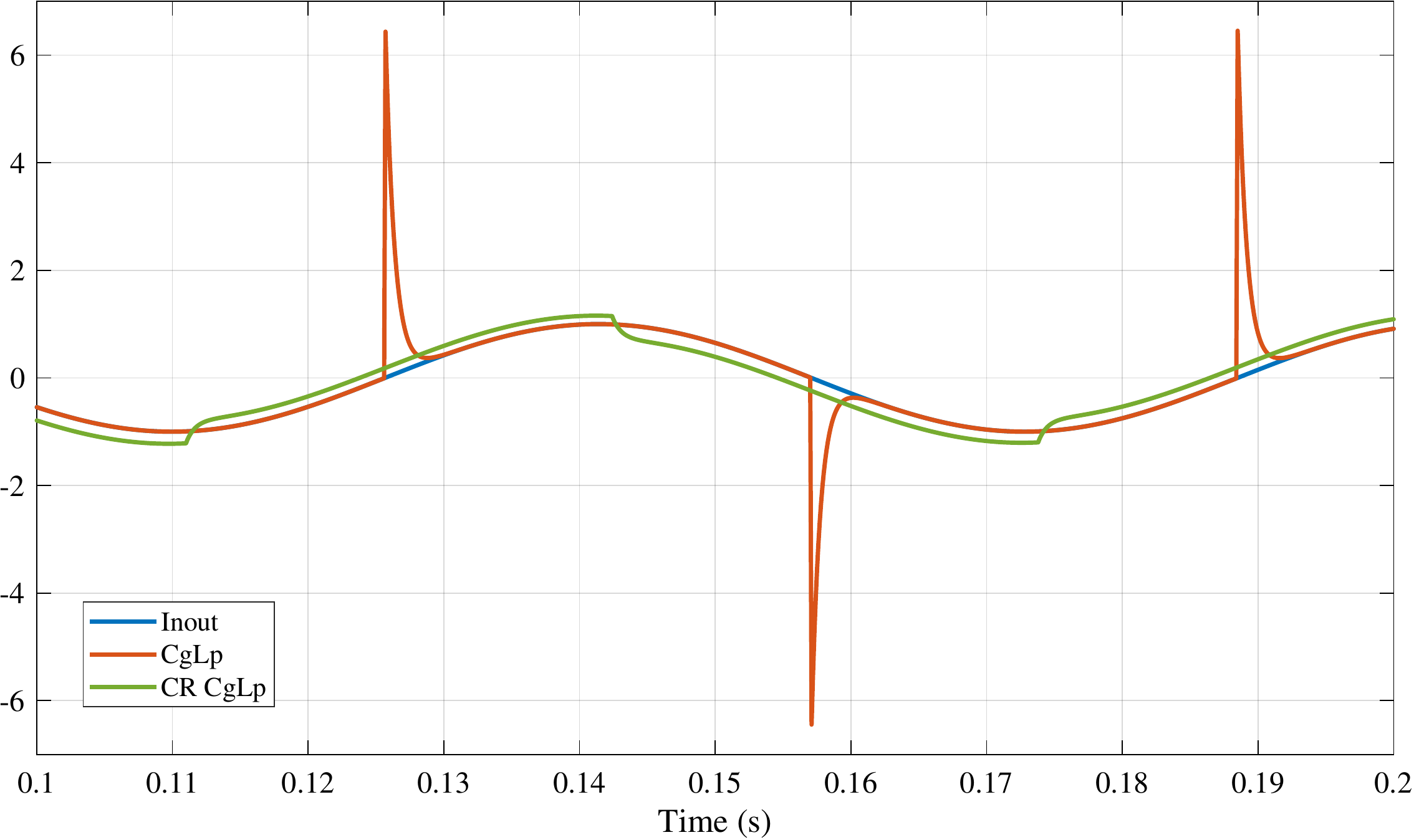}%
		\caption{Sinusoidal response of conventional CgLp compared to CR CgLp. Input is $\sin(100t)$..}
		\label{fig:cglp_sin}
	\end{subfigure}
	\caption{Simulation results for sinusoidal response of CI and CgLp compared to their CR architecture proposed in this paper. }
	\label{fig:sin_clegg_cglp}
\end{figure}   
\begin{pro}\label{rem:df}
	For $\omega_h = \infty$, the CR architecture has the same DF as the ${\sum{}}_R$ alone. 
\end{pro}
\begin{proof}
	Let the states in CR architecture be denoted as shown in Fig.~\mbox{\ref{fig:new_arch}}. For the purpose of DF and HOSIDF analysis, one should have $e(t)=\sin(\omega t)$. Obviously, the steady-state response of $x_1(t)$ is:
	\begin{equation}
		x_1(t)=l(\omega)\sin(\omega t+\psi(\omega))
	\end{equation} 
	where $l(\omega)=\vert L(j\omega)\vert$ and $\psi(\omega)=\angle L(j\omega)$. Considering $x_1(t)$ as the input to the reset element and according to Eq.~\mbox{\eqref{eq:hosidf}},
	\begin{equation}
		x_{21}(t)=h_1(\omega)l(\omega)\sin(\omega t+\theta_1(\omega))
	\end{equation}
	where $x_{21}(t)$ stands for first harmonic of $x_2(t)$ and $h_1(\omega)=\vert H_1(j\omega)\vert$ and $\theta_1(\omega)=\psi(\omega)+\angle H_1(j\omega)$. And lastly,
	\begin{equation}
		u_{1}(t)=h_1(\omega)l(\omega)r(\omega)\sin(\omega t+\varphi_1(\omega))
	\end{equation}
	where $u_{1}(t)$ stands for first harmonic of $u(t)$ and $r(t)=\vert R(j\omega)\vert$ and $\varphi_1(\omega)=\theta_1(\omega)+\angle R(j\omega)$.
	Since $R(j\omega) \approx L^{-1}(j\omega)$ for $\omega \ll \omega_h$, it can be seen that
	\begin{align}
		\vert u_1(t)\vert&=h_1(\omega),\\
		\varphi_1(\omega)&=\psi(\omega)+\angle H_1(j\omega)-\psi(\omega)= \angle H_1(j\omega).
	\end{align}
\end{proof} 
\begin{pro}\label{rem:hosidf}
	The magnitude of higher-order harmonics for CR architecture is reduced compared to the ${\sum{}}_R$ alone. 
\end{pro}
\begin{proof}
	Following the same reasoning as Proposition~\mbox{\ref{rem:df}}, one has 
	\begin{equation}
		u_n(t)= { h_{n}( \omega ) l( \omega )} { r( n\omega )}\sin(n\omega t + \varphi_n(\omega)),
	\end{equation}
	where $u_n(t)$ is the $n^\text{th}$ harmonic of $u(t)$, $h_n(\omega)=\vert H_n(j\omega)\vert $ and $\varphi_n= \angle  H_n(j\omega)+\angle  L(j\omega) + \angle R(jn\omega)$.
	Since $r^{-1}(\omega) \approx l(\omega) $ for $\omega_h \gg \omega_l$, and since $l(\omega)$ is an increasing function 
	\begin{equation}
		A_n(\omega)< h_{n}( \omega ),
	\end{equation}
	where $A_n(\omega)$ stands for $\vert u_n(t)\vert$. In other terms, for large enough $\omega_{h}$,
	\begin{equation}
		A_n(\omega)\approx \sqrt{\frac{(\omega/\omega_l)^2+1}{(n\omega/\omega_l)^2+1}}h_n(\omega).
	\end{equation}
	For $\omega\ll\omega_{l}$,
	\begin{equation}
		A_n(\omega)=h_n(\omega)
	\end{equation}
	and for $\omega\gg\omega_{l}$,
	\begin{equation}
		A_n(\omega)=\frac{1}{n}h_n(\omega).
	\end{equation}
\end{proof}
Fig.~\mbox{\ref{fig:HOSIDF}} illustrates the harmonic generation for CR architecture.\\
Theorem~\ref{thm:cont} and Propositions~\ref{rem:df} and~\ref{rem:hosidf} may seem somewhat trivial, however they indicate very important features of the CR architecture in terms of steady-state performance. As mentioned earlier the frequency domain analysis and design for reset control systems heavily depends on the accuracy of DF approximation. The CR architecture maintains the DF characteristics of the reset elements and reduces the higher-order harmonics which makes the DF approximation more accurate. It is shown in~\cite{karbasizadeh2020benefiting, karbasizadeh2021fractional} that it improves the performance of the systems in terms of steady-state precision. \\
Moreover, the discontinuity of output signal in reset controllers creates practical problems such as amplifier or actuator saturation and excitation of higher frequency modes for complex plants. The CR architecture will solve these problems by reducing the known peaks in the control input of the reset control systems. \\
In order to illustrate the  effect of the CR architecture on HOSIDF of reset elements, the HOSIDF of a Clegg Integrator (CI) and a CR CI are compared in Fig.~\ref{fig:clegg_hosidf}, this figure shows that while the DF for these two elements are identical a significant reduction in HOSIDF of CR CI with respect to CI happens, this indicates that as we approach higher frequencies, the DF will become a more accurate approximation in CR CI. The same comparison is made for CgLp and CR CgLp in Fig.~\ref{fig:cglp_hosidf}. Both CgLps are designed to create a phase lead of $15^\circ$ at 100 rad$/$s while maintaining a constant gain. A significant reduction in magnitude of higher-order harmonics is also clear here, which indicates that CR CgLp has a much closer behaviour to the first-order harmonic which is the ideal behavior for reset control systems.\\
In Fig.~\ref{fig:sin_clegg_cglp}, the sinusoidal response of CI vs. CR CI at 1~rad$/$s and CgLp vs. CR CgLp at 100~rad$/$s are depicted. At both comparisons, it is clear that the output of CR architecture is continuous as opposed to reset elements in their conventional form, and the response are much smoother which shows the reduction of higher-order harmonics. It has to be noted for the case of CgLp, the big peak in the response, which can cause aforementioned practical issue, is removed in CR CgLp.\\
The superiority of CgLp control structures over other reset control strategies in precision motion control has been shown in many researches~\cite{saikumar2019constant,dastjerdi2021frequency,karbasizadeh2021fractional}. In the remainder of this paper, for the sake of conciseness, only CR CgLp architecture will be studied. However, the same approach can be used for other reset control structures.\\
For the case of CR CgLp, the magnitude of higher-order harmonics for frequencies lower than $\omega_{c}$ (where it matters the most for tracking and disturbance rejection~\cite{karbasizadeh2020benefiting,karbasizadeh2021fractional}) is also affected by parameters other than $\omega_l$. These parameters are $\omega_r$ and $\gamma$. However, unlike $\omega_{l}$, these two parameters also affect the DF phase and consequently the amount phase lead created by CR CgLp. This creates a trade-off between reduction of higher-order harmonics magnitude and maximum achievable Phase Advantage (PA) of CR CgLp. Fig.~\ref{fig:wr_sweep} illustrates the trade-off. CgLp will be logically designed to provide phase lead at cross-over frequency, i.e., $\omega_c$. As $\omega_{r}$ approaches $\omega_{c}$ the integral of $3^\text{rd}$ harmonic magnitude over frequencies below $\omega_{c}$ decreases significantly. The reduction of integral value is an indication of the reduction of magnitude of higher-order harmonics in general. Furthermore, the peak of higher-order harmonics will also shift to higher frequencies when $\omega_r$ approaches higher frequencies. Thus it seems logical to have this peak in frequencies where tracking and disturbance rejection performance is not a matter of concern, i.e., the frequencies after the bandwidth. When $\omega_{r}$ is in $[\omega_{c}, 1.5\omega_{c}]$, higher-order harmonics are very low and still a PA up to 35$^\circ$ is achievable. This can be a general guideline for tuning $\omega_{r}$ in CR CgLp.
\begin{figure}[t!]
	\centering
	\includegraphics[width=\textwidth]{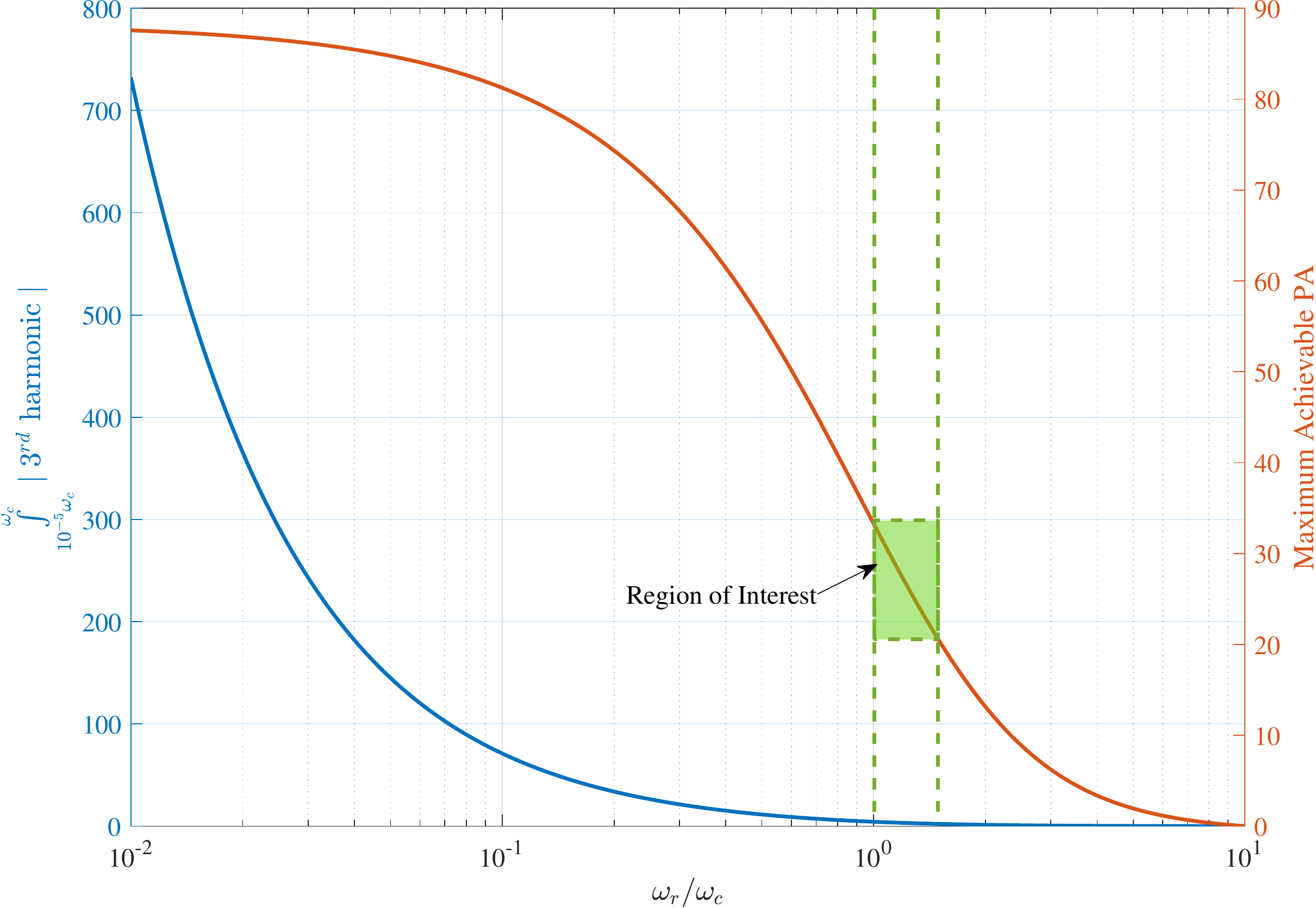}
	\caption{Integral of 3$^\text{rd}$ harmonic magnitude for frequencies below $\omega_{c}$ and the maximum achievable PA at $\omega_{c}$ vs. the ratio of $\omega_{r}$ to $\omega_{c}$. $\gamma=-1$. }
	\label{fig:wr_sweep}
\end{figure}
\section{Closed-Loop Transient Response Properties of the CR CgLp Architecture}

In the researches done on CgLp control systems in the literature, the only considered design parameter for changing the transient response of systems is phase margin. In the context of linear control systems, phase margin is determining parameter; however, that is not the case for reset control system and especially for the CR architecture presented in this paper. Referring to Eq.~\eqref{eq:reset_law}, speaking in terms of the closed loop, in CR architecture, the reset condition is not only based on the error signal but a linear combination of error and its derivative. This will change transient response of the system as well~\cite{dastjerdi2020optimal,vanden2020hybrid,cai2020optimal}. In order to study the effect of parameters of CR architecture on transient response of a closed-loop precision motion control system, a data-based approach has been used in this paper.\\ Fig.~\ref{fig:block_closed} shows the block diagram of the control loop. As it is shown in the figure, the reset part of CgLp, i.e., ${\sum}_{R}$, is surrounded by $L(s)$ and $R(s)$ to create a CR CgLp.\\
Following the discussion in Section~\mbox{\ref{sec:dynamics}}, the plant which is used for this data-based study is a mass system, i.e., $P(s)=1/s^2$. In experimental validation, it will be shown that the analysis will also hold mass-spring-damper systems.\\
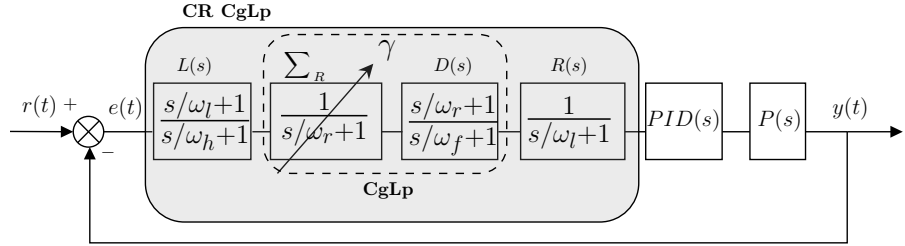
\begin{figure}[t!]
	\centering
	\resizebox{\textwidth}{!}{

		\tikzset{every picture/.style={line width=0.75pt}} %set default line width to 0.75pt        
		
		\begin{tikzpicture}[x=0.75pt,y=0.75pt,yscale=-1,xscale=1]
			%uncomment if require: \path (0,339); %set diagram left start at 0, and has height of 339
			
			%Shape: Rectangle [id:dp5398233698731352] 
			\draw  [line width=0.75]  (100.9,70) -- (170.9,70) -- (170.9,125) -- (100.9,125) -- cycle ;
			%Shape: Rectangle [id:dp43816118744263943] 
			\draw  [line width=0.75]  (185,70) -- (265,70) -- (265,125) -- (185,125) -- cycle ;
			%Straight Lines [id:da8490259811573035] 
			\draw [line width=0.75]    (190,135) -- (257.03,58.26) ;
			\draw [shift={(259,56)}, rotate = 491.13] [fill={rgb, 255:red, 0; green, 0; blue, 0 }  ][line width=0.08]  [draw opacity=0] (10.72,-5.15) -- (0,0) -- (10.72,5.15) -- (7.12,0) -- cycle    ;
			%Shape: Rectangle [id:dp4354848591777598] 
			\draw  [line width=0.75]  (530,70) -- (570,70) -- (570,125) -- (530,125) -- cycle ;
			%Straight Lines [id:da5399269145470293] 
			\draw [line width=0.75]    (65,105) -- (100,105) ;
			%Straight Lines [id:da31293197971436437] 
			\draw [line width=0.75]    (600,105) -- (600,185) -- (55,185) -- (55,119) ;
			\draw [shift={(55,116)}, rotate = 450] [fill={rgb, 255:red, 0; green, 0; blue, 0 }  ][line width=0.08]  [draw opacity=0] (8.93,-4.29) -- (0,0) -- (8.93,4.29) -- cycle    ;
			%Straight Lines [id:da9287897608622762] 
			\draw [line width=0.75]    (-2.42,104.78) -- (40.08,104.99) ;
			\draw [shift={(43.08,105)}, rotate = 180.27] [fill={rgb, 255:red, 0; green, 0; blue, 0 }  ][line width=0.08]  [draw opacity=0] (8.93,-4.29) -- (0,0) -- (8.93,4.29) -- cycle    ;
			%Straight Lines [id:da5744123676950859] 
			\draw [line width=0.75]    (600,105) -- (637,105) ;
			\draw [shift={(640,105)}, rotate = 180] [fill={rgb, 255:red, 0; green, 0; blue, 0 }  ][line width=0.08]  [draw opacity=0] (8.93,-4.29) -- (0,0) -- (8.93,4.29) -- cycle    ;
			%Straight Lines [id:da4650808618283646] 
			\draw [line width=0.75]    (172,105) -- (185,105) ;
			%Straight Lines [id:da3258710761711059] 
			\draw [line width=0.75]    (440,105) -- (455,105) ;
			%Shape: Rectangle [id:dp2352161520774516] 
			\draw  [line width=0.75]  (455,70) -- (510,70) -- (510,125) -- (455,125) -- cycle ;
			%Straight Lines [id:da6356652741700044] 
			\draw [line width=0.75]    (510,105) -- (530,105) ;
			%Flowchart: Summing Junction [id:dp013541175974679964] 
			\draw   (43.08,105) .. controls (43.08,98.89) and (47.98,93.94) .. (54.04,93.94) .. controls (60.09,93.94) and (65,98.89) .. (65,105) .. controls (65,111.11) and (60.09,116.06) .. (54.04,116.06) .. controls (47.98,116.06) and (43.08,111.11) .. (43.08,105) -- cycle ; \draw   (46.29,97.18) -- (61.79,112.82) ; \draw   (61.79,97.18) -- (46.29,112.82) ;
			%Shape: Rectangle [id:dp9854798752451399] 
			\draw  [line width=0.75]  (365,70) -- (440,70) -- (440,125) -- (365,125) -- cycle ;
			%Shape: Rectangle [id:dp9348131494781597] 
			\draw  [line width=0.75]  (279.5,70) -- (348.5,70) -- (348.5,125) -- (279.5,125) -- cycle ;
			%Straight Lines [id:da13153924917245696] 
			\draw [line width=0.75]    (265,105) -- (280,105) ;
			%Straight Lines [id:da6177832264293783] 
			\draw [line width=0.75]    (349,105) -- (365,105) ;
			%Straight Lines [id:da3570263221656147] 
			\draw [line width=0.75]    (570,105) -- (600,105) ;
			%Rounded Rect [id:dp07250460123358793] 
			\draw  [fill={rgb, 255:red, 155; green, 155; blue, 155 }  ,fill opacity=0.2 ] (95,58) .. controls (95,42.54) and (107.54,30) .. (123,30) -- (422,30) .. controls (437.46,30) and (450,42.54) .. (450,58) -- (450,142) .. controls (450,157.46) and (437.46,170) .. (422,170) -- (123,170) .. controls (107.54,170) and (95,157.46) .. (95,142) -- cycle ;
			%Rounded Rect [id:dp11651876027252994] 
			\draw  [dash pattern={on 4.5pt off 4.5pt}] (180,55) .. controls (180,43.95) and (188.95,35) .. (200,35) -- (335,35) .. controls (346.05,35) and (355,43.95) .. (355,55) -- (355,115) .. controls (355,126.05) and (346.05,135) .. (335,135) -- (200,135) .. controls (188.95,135) and (180,126.05) .. (180,115) -- cycle ;

			% Text Node
			\draw (19,87) node  [font=\large]  {$r( t)$};
			% Text Node
			\draw (602.5,88) node  [font=\large]  {$y( t)$};
			% Text Node
			\draw (81,88) node  [font=\large]  {$e( t)$};
			% Text Node
			\draw (41,85) node    {$+$};
			% Text Node
			\draw (68,120) node    {$-$};
			% Text Node
			\draw (189.5,75.57) node [anchor=north west][inner sep=0.75pt]  [font=\huge]   {$\frac{1}{s/\omega _{r } +1}$};
			% Text Node
			\draw (367,78) node [anchor=north west][inner sep=0.75pt]  [font=\huge]  {$\frac{1}{s/\omega _{l} +1}$};
			% Text Node
			\draw (102.9,74.65) node [anchor=north west][inner sep=0.75pt]  [font=\huge]  {$\frac{s/\omega _{l} +1}{s/\omega _{h} +1}$};
			% Text Node
			\draw (261,37.4) node [anchor=north west][inner sep=0.75pt]  [font=\LARGE]  {$\gamma $};
			% Text Node
			\draw (281.5,75.57) node [anchor=north west][inner sep=0.75pt]  [font=\huge]  {$\frac{s/\omega _{r} +1}{s/\omega _{f} +1}$};
			% Text Node
			\draw (454,87.4) node [anchor=north west][inner sep=0.75pt]  [font=\large]  {$PID( s)$};
			% Text Node
			\draw (534,87.4) node [anchor=north west][inner sep=0.75pt]  [font=\large]  {$P( s)$};
			% Text Node
			\draw (120,12) node [anchor=north west][inner sep=0.75pt]   [align=left] {\textbf{CR CgLp}};
			% Text Node
			\draw (250,140) node [anchor=north west][inner sep=0.75pt]   [align=left] {\textbf{CgLp}};
			% Text Node
			\draw (193,45) node [anchor=north west][inner sep=0.75pt] [font=\normalsize]   {${\displaystyle \sum{}}_{R}$};
			% Text Node
			\draw (385,50) node [anchor=north west][inner sep=0.75pt]  [font=\normalsize]  {$R( s)$};
			% Text Node
			\draw (116,50) node [anchor=north west][inner sep=0.75pt]  [font=\normalsize]  {$L( s)$};
			% Text Node
			\draw (301,50) node [anchor=north west][inner sep=0.75pt]  [font=\normalsize]  {$D( s)$};

		\end{tikzpicture}
	}
	\caption{The control loop used for precision motion control using CR CgLp. $P(s)$ is the plant. $PID(s)=k_p\left(1+{\omega_{i}}/{s}\right)\left(\frac{{s}/{\omega_d}+1}{{s}/{\omega_t}+1}\right)$.}
	\label{fig:block_closed}
\end{figure}
The $H_\beta$ condition for stability of the reset control systems necessarily requires the BLS to be stable. Thus, a PID controller is present in the loop. However, according to loop-shaping technique, to ensure the maximum steady-state precision performance for the system, the differentiation part of the PID should be as weak as possible to only guarantee the stability of the BLS. Normally, such a tuning for PID control system will perform poorly in terms transient response in absence of CR CgLp. Nevertheless, it will be shown that the presence CR CgLp will significantly improve transient response without affecting the maximally precise steady-state performance of the system. In this study, using a rule of thumb, the PID is tuned such that the BLS has $5^\circ$ phase margin, which is enough to stabilize the BLS and since it has a weak differentiator, does not jeopardize the steady-state precision. The following  equation shows the parameters chosen in this regard.\\
\begin{equation}
	\omega_{i}=\omega_{c}/10,\quad\omega_d=\omega_{c}/1.2,\quad\omega_t=1.2\omega_c
\end{equation} 
And consequently, $k_p$ can be determined according to $\omega_c$. According to the discussions in Section~\ref{sec:steady-state}, without loss of generality, for this data-based study,
\begin{equation}
	\omega_{r}=1.2\omega_c.
\end{equation} 
This leaves the effect of $\gamma$ and $\omega_{l}$ to be studied. Since $\omega_{r}$ and the parameters of PID are fixed, the only parameter which affects the phase margin of the designed system is $\gamma$. It has to be noted, that according to Proposition~\ref{rem:df}, CR architecture does not change the DF, thus $\omega_{l}$ does not have an effect on phase margin. Fig.~\ref{fig:hosidf_gamma} shows the open-loop DF of the system under study and also the effect of $\gamma$ on phase margin. $\gamma=1$ indicates the base linear system and as the value $\gamma$ decreases the phase margin will increase. At $\omega=\omega_{c}$, it can be seen that CR CgLp not only does not change the gain behavoiur, but also creates a positive slope in phase, which resembles the complex-order controllers. In the following, the effect of phase margin and $\omega_{l}$ on overshoot and settling time of the closed-loop system will be shown.
\begin{figure}[t!]
	\centering
	\includegraphics[width=\textwidth]{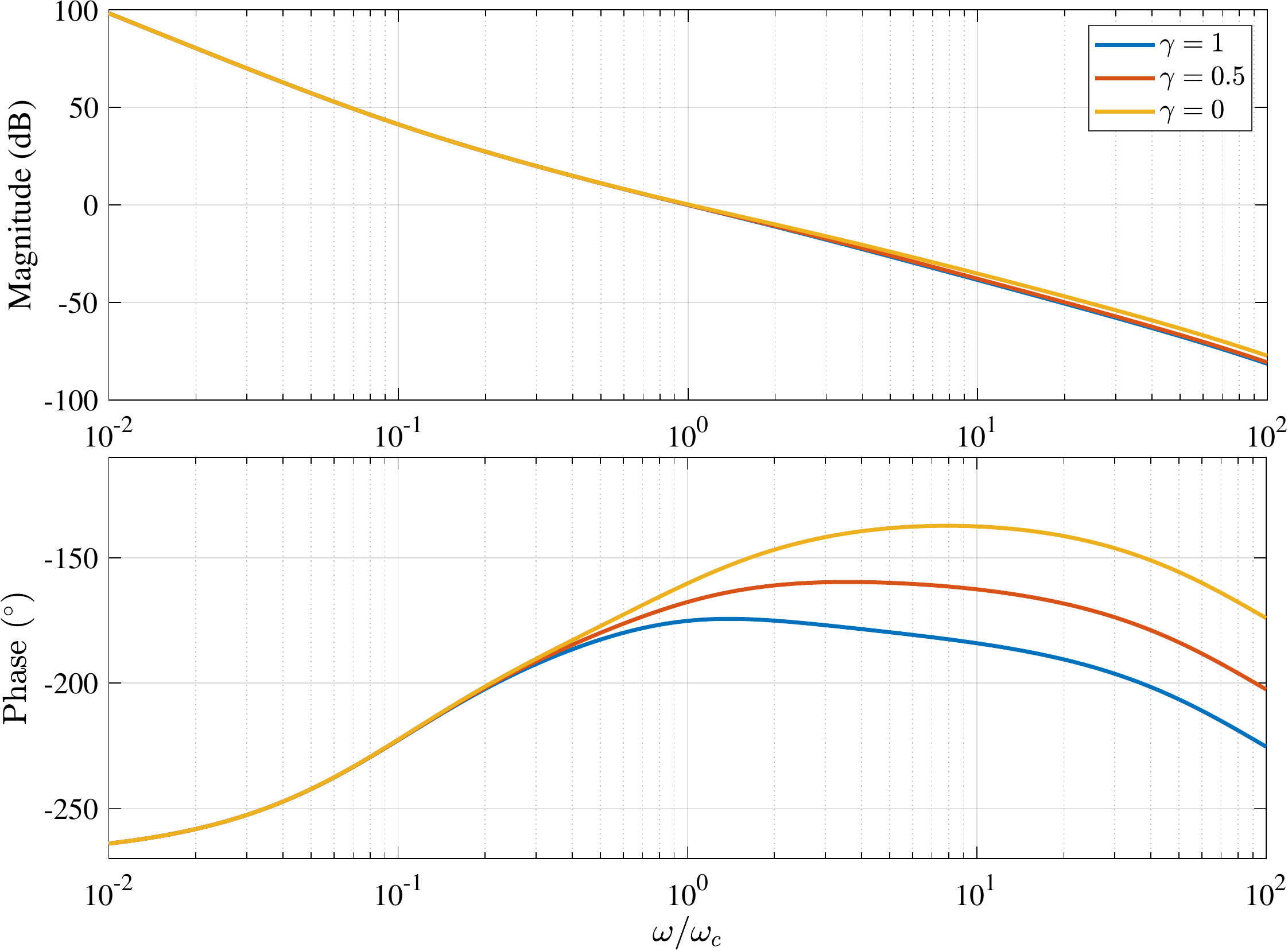}
	\caption{DF of the open loop system for different $\gamma$ values versus the ration of $\omega/\omega_{c}$.}
	\label{fig:hosidf_gamma}
\end{figure}
\subsection{Overshoot}
As mentioned before, it is expected that the variation of phase margin caused by variation of $\gamma$ and the variation on $\omega_{l}$ create different transient responses for the closed-loop system. In order to do a data-based study, a unit step reference was given to the closed-loop system and the the response was simulated using Simulink environment of Matlab. The overshoot versus the variation of $\omega_{l}$ and phase margin is depicted in Fig.~\ref{fig:overshoot}.\\
From Fig.~\ref{fig:overshoot}, it can be concluded that similar to linear controllers, with increase of the phase margin the overshoot decreases almost linearly. Furthermore, for a constant value of phase margin as $\omega_{l}$ decrease the overshoot decreases and for some configurations a non-overshoot performance is realizable. It should be also noted that as $\omega_{l}$ increases, it weakens the lead element $L(s)$ and thus system gradually tends to the performance of the conventional CgLp. Overshoot of the system in the absence of CR~CgLp, i.e., BLS, is 96\%.\\
In the range of Phase Margin (PM) $\in \left[10,30\right]$ and $\omega_{l}/\omega_c \in \left[0.1,1\right]$, the decrease of overshoot (OS) is almost linear with respect decrease of $\log(\omega_{l})$. A fitting operation reveals the following relation between the OS and PM and $\omega_{l}$.  
\begin{equation}\label{eq:fit_overshoot}
	OS=0.95\log\left(\frac{\omega_{l}}{\omega_{c}}\right)-0.04PM+1.25
\end{equation}
where PM is in degrees. \\
\begin{figure}[t!]
	\centering
	\includegraphics[width=\textwidth]{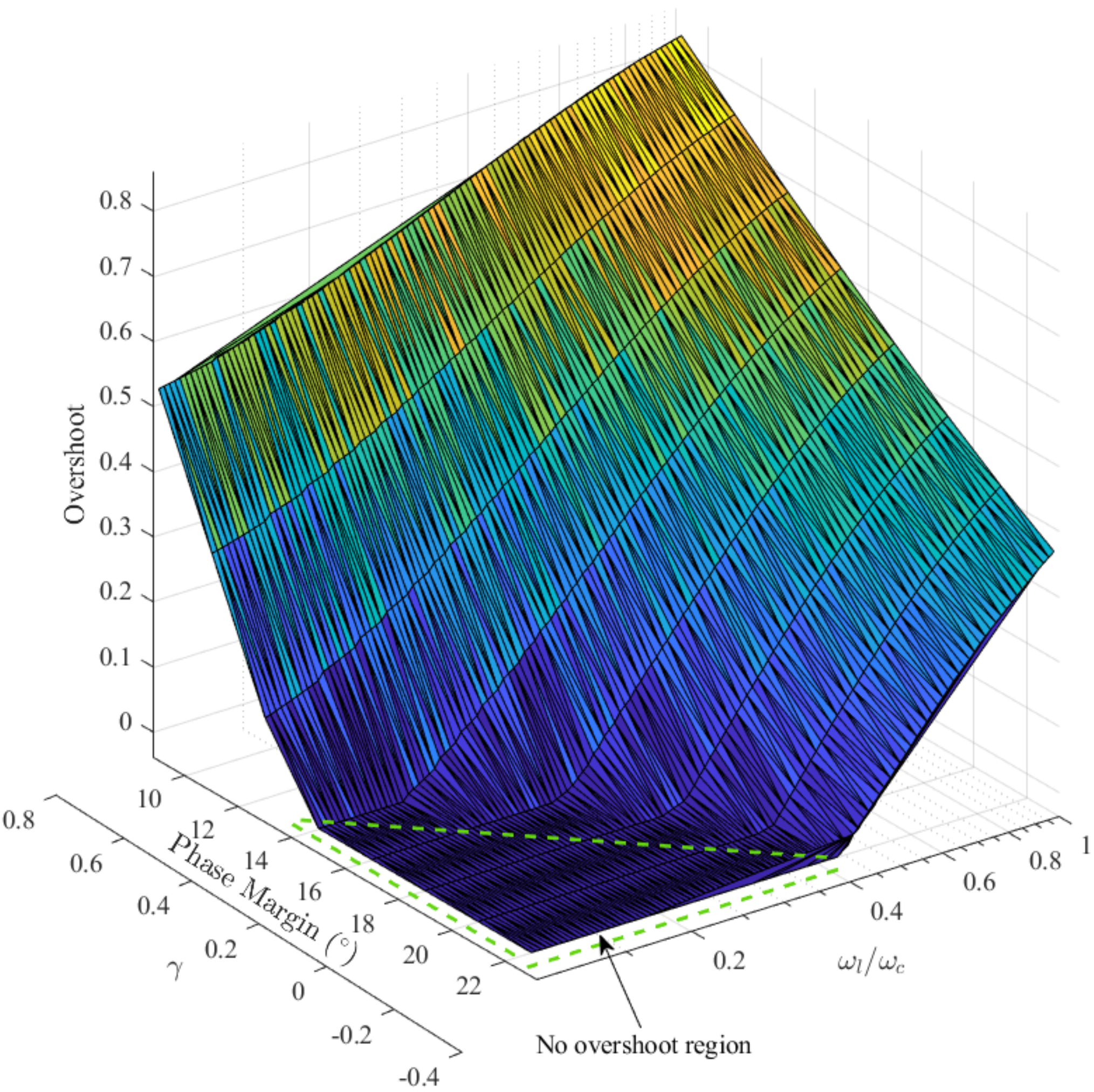}
	\caption{The overshoot of the system to a unit step for phase margin in range of $\left[5,22\right]$ and $\omega_l/\omega_c \in \left[0.1,1\right]$. $5^\circ$ of the phase margin is provided through base linear system. The overshoot in the absence of the CR~CgLp, i.e., BLS, is 0.962.}
	\label{fig:overshoot}
\end{figure}
In order to better illustrate the effect of these two parameters on overshoot and in general transient response of the closed-loop system, one can refer to Fig.~\ref{fig:step_time_sweep_PM_wl}. For this simulation $\omega_{c}=100~\text{rad}/\text{s}$. Fig.~\ref{fig:step_time_sweep_wl} shows the reduction of overshoot by reduction of $\omega_{l}$, the non-overshoot response is shown to be realizable. However, too much reduction of $\omega_{l}$ can result in long settling times as is the case for $\omega_{l}=10~\text{rad}/\text{s}$. Obviously, since CgLp does not contain $\omega_{l}$, it has only one response.\\
Fig.~\ref{fig:step_time_sweep_PM} demonstrates the effect of PM on step response of the system while $\omega_{l}=33~\text{rad}/\text{s}$, the presence of CR architecture amplifies the reduction of overshoot caused by increase of PM. It has to be noted that various values of PM is achieved by changing $\gamma$.\\
The study shows the significant improve in transient response by CR CgLp. It worth mentioning that it will be showed later that this improvement in transient will not sacrifice the steady-state response.\\  
\begin{figure}[!t]
	\centering
	\begin{subfigure}{\textwidth}
		\includegraphics[width=\textwidth]{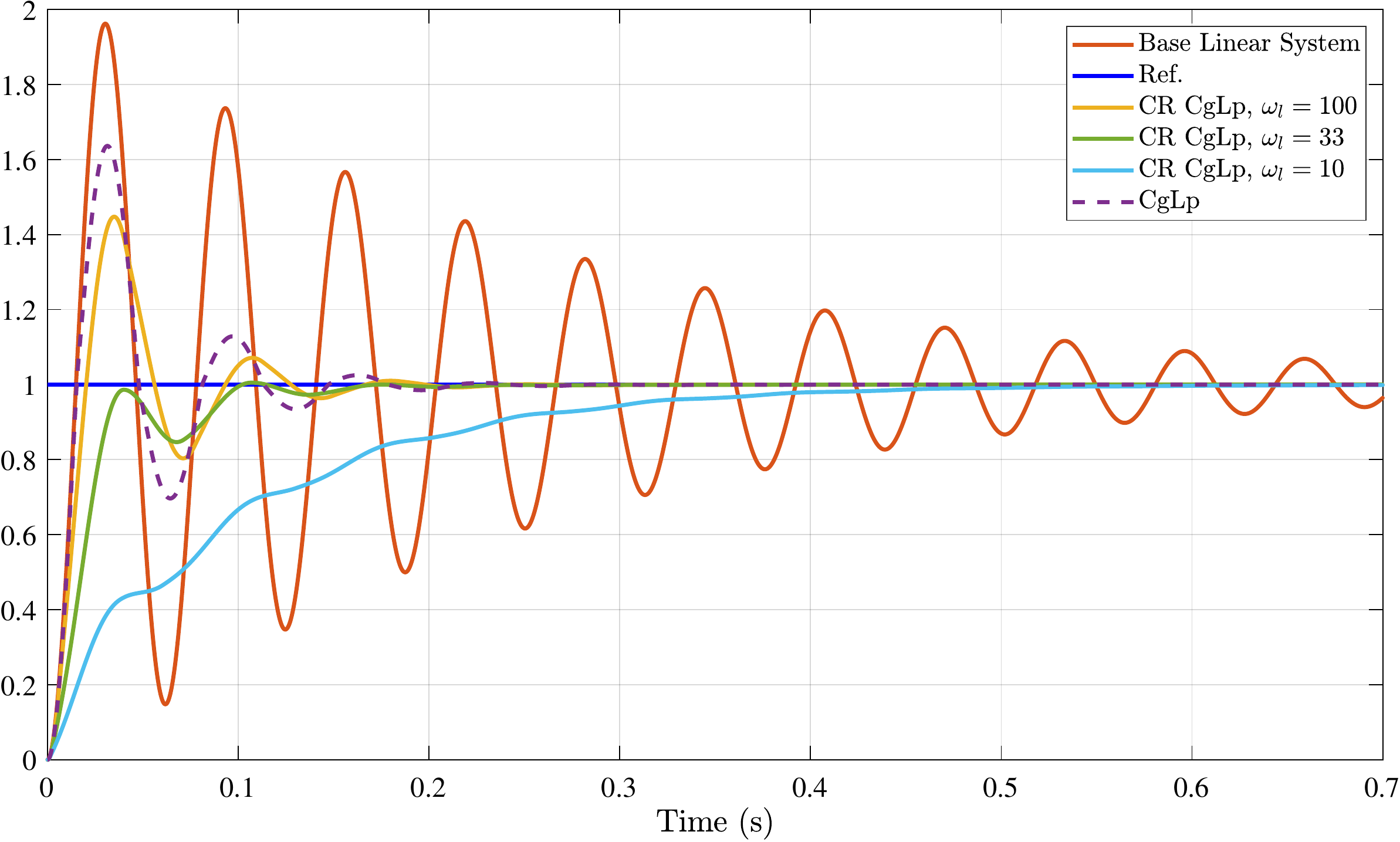}%
		\caption{Step response of closed-loop system for base linear system, CgLp and CR~CgLp for various values of $\omega_{l}$ in $\text{rad}/\text{s}$. PM is fixed at $20^\circ$ and $\omega_{c}=100~\text{rad}/\text{s}$. }
			\label{fig:step_time_sweep_wl}
	\end{subfigure}\\
	\begin{subfigure}{\textwidth}
		\includegraphics[width=\textwidth]{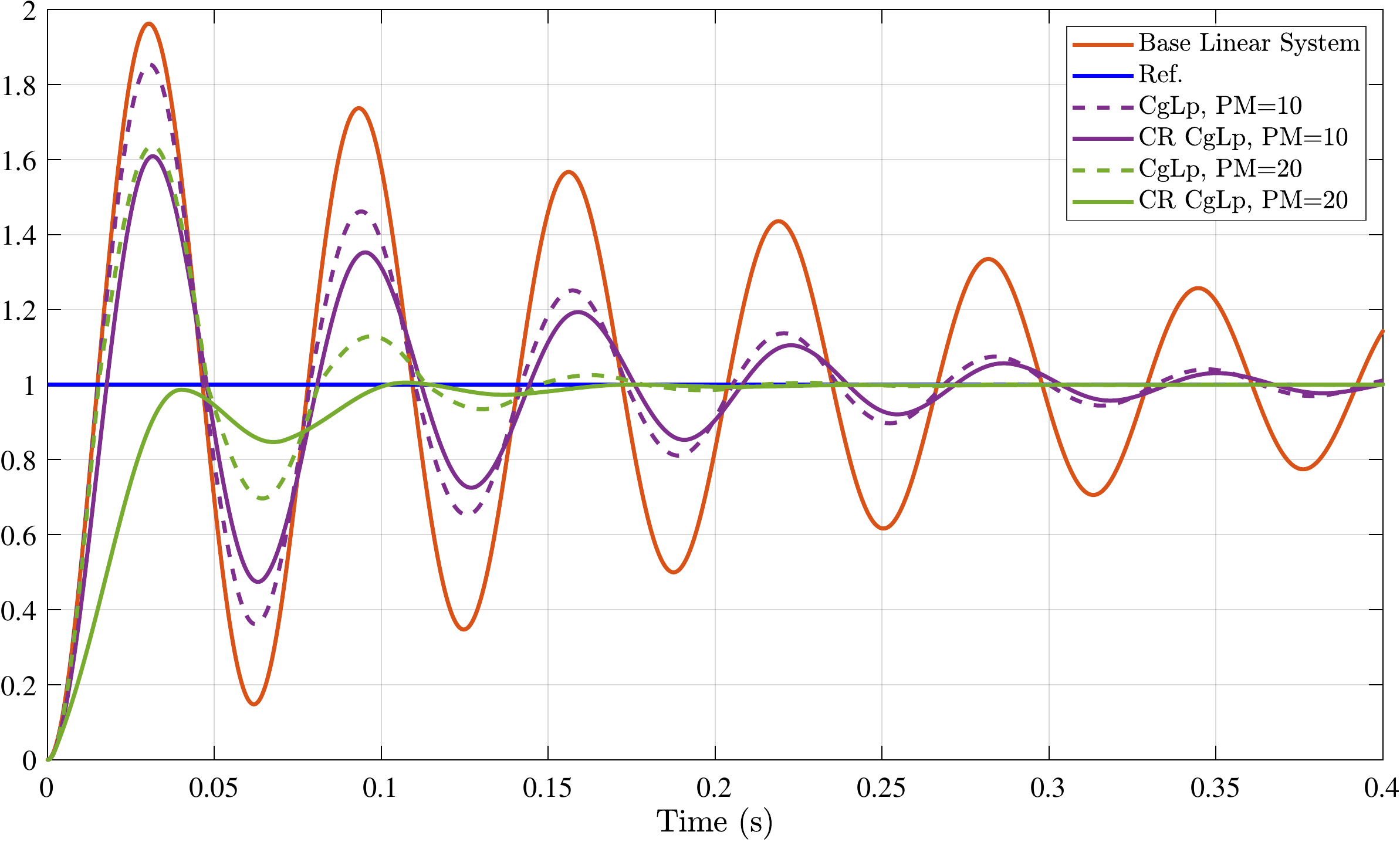}%
			\caption{Step response of closed-loop system for base linear system, CgLp and CR~CgLp for various values of PM. $\omega_{c}=33~\text{rad}/\text{s}$ and $\omega_{c}=100~\text{rad}/\text{s}$.}
			\label{fig:step_time_sweep_PM}
	\end{subfigure}

	\caption{Step response of closed-loop system for base linear system, CgLp and CR~CgLp for various values of PM and $\omega_l$.}
	\label{fig:step_time_sweep_PM_wl}
\end{figure} 
\subsection{Settling time}  
According to Fig.~\ref{fig:overshoot} and~\ref{fig:step_time_sweep_wl}, reduction of $\omega_{l}$ generally decreases overshoot, it may have an adverse effect on settling time. In order to find a sweet spot where overshoot and settling time are improved simultaneously the same sweep as Fig.~\ref{fig:overshoot} has been done for settling time and depicted in Fig.~\ref{fig:settling_time}. According to this figure, for a constant $\omega_{l}$ the settling time decreases with increase of PM as as like the case for linear controllers. However, there is no linear relation for $\omega_{l}/\omega_{c}$ and settling time.\\
As a rule of thumb, $\omega_{l}/\omega_{c} \in \left[0.3,0.6\right]$ and PM larger than $20^\circ$ shows a favorable settling time. In this range the settling time of the CR~CgLp is shorter than CgLp and referring to Fig.~\ref{fig:overshoot}, non-overshoot performance can also be achieved. Thus one can use this general rule of thumb as the tuning guideline of CR~CgLp. 
\begin{figure}[t!]
	\centering
	\includegraphics[width=\textwidth]{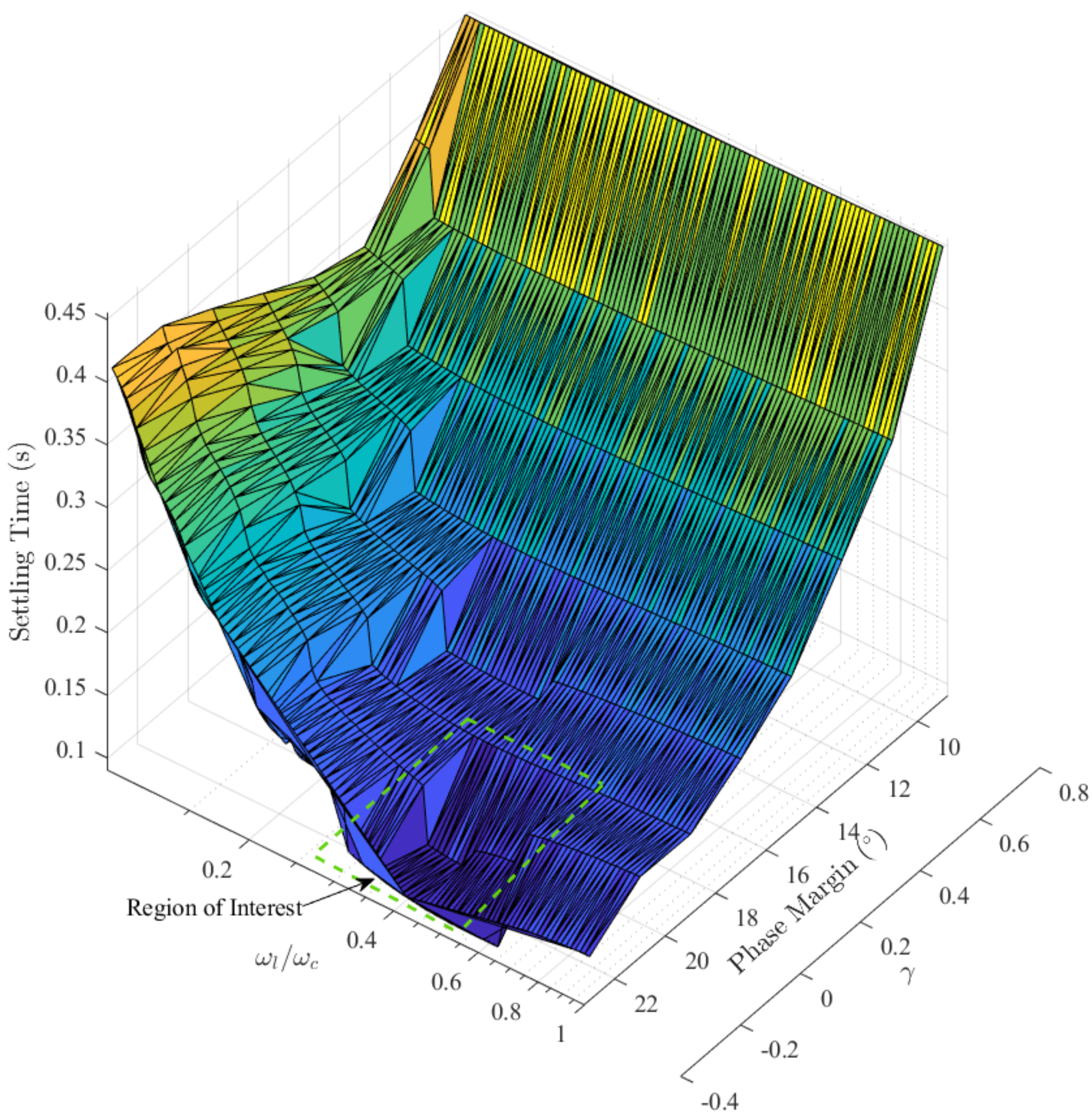}
	\caption{The settling time of the system for a unit step for phase margin in range of $\left[5,22\right]$ and $\omega_l/\omega_c \in \left[0.1,1\right]$. $\omega_{c}=100~\text{rad}/\text{s}$. The settling time in the absence of the CR~CgLp, i.e., BLS, is 0.945 s.}
	\label{fig:settling_time}
\end{figure}
\section{Closed-Loop Steady-State Performance of the CR CgLp Architecture}
As discussed earlier, the DF method can be used as an approximation for open-loop steady-state performance of reset control systems. The DF can also be used to find the sensitivity functions of closed-loop reset control systems using the linear relations between open-loop transfer functions and closed-loop sensitivity functions. While the resulted sensitivity plots show the ideal steady-state behaviour for the designed reset controllers, the presence of higher-order harmonics makes achieving it impossible. Thus, as discussed in Section~\mbox{\ref{sec:steady-state}}, reducing higher-order harmonics brings the reset controller closer to the ideal behaviour.\\
It is shown that CR architecture and its tuning guidelines can reduce the magnitude of higher-order harmonics. Thus, it is expected that actual closed-loop steady-sate performance is very close to approximation created by DF. In order to verify the latter, a comparison has been made. A series of simulations has been run to determine the actual sensitivity functions values for different frequencies. However, because of nonlinearity of the system, the output will not be sinusoidal. To approximate, the second norm of the signals has been used.\\ 
%\begin{figure}[!t]
%	\centering
%	\subfloat[Sensitivity plot for BLS ($S_\text{BLS}$) along sensitivity for reset control systems calculated based on DF ($S_\text{DF}$), and sensitivity calculated  based on infinity norm, i.e., $\frac{\Vert y(t) \Vert_\infty}{\Vert r(t) \Vert_\infty}$ for CgLp and CR CgLp ($S_\text{CgLp}$ and $S_\text{CR\, CgLp}$). ]{\includegraphics[width=\textwidth]{sensitivity_sim-c.pdf}%
%		\label{fig:step_time_sweep_wl}}
%	\hfil
%	\subfloat[Complementry sensitivity plot for BLS ($T_\text{BLS}$) along complementry sensitivity for reset control systems calculated based on DF ($T_\text{DF}$), and complementry sensitivity calculated  based on infinity norm, i.e., $\frac{\Vert y(t) \Vert_\infty}{\Vert r(t) \Vert_\infty}$ for CgLp and CR CgLp ($S_\text{CgLp}$ and $S_\text{CR\, CgLp}$).]{\includegraphics[width=\textwidth]{comp_sensitivity_sim-c.pdf}%
%		\label{fig:step_time_sweep_PM}}
%	\caption{Step response of closed-loop system for base linear system, CgLp and CR~CgLp for various values of PM and $\omega_l$.}
%	\label{fig:step_time_sweep_PM_wl}
%\end{figure} 
\begin{figure}[t!]
	\centering
	\includegraphics[width=\textwidth]{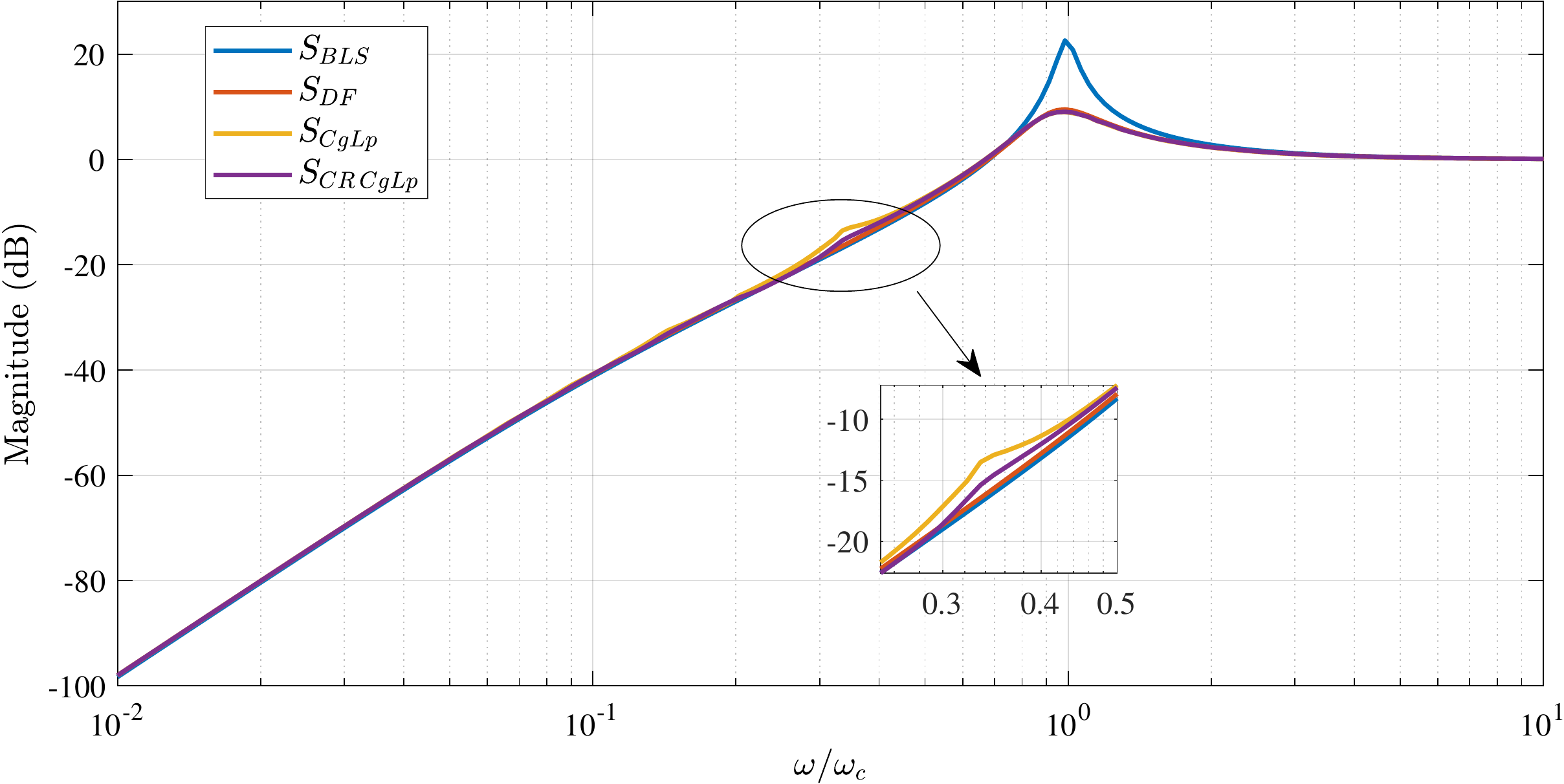}
	\caption{Sensitivity plot for BLS ($S_\text{BLS}$) along sensitivity for reset control systems calculated based on DF ($S_\text{DF}$), and sensitivity calculated  based on infinity norm, i.e., $\frac{\Vert e(t) \Vert_2}{\Vert r(t) \Vert_2}$ for CgLp and CR CgLp ($S_\text{CgLp}$ and $S_\text{CR\, CgLp}$).}
	\label{fig:sensitivity_sim}
\end{figure}
According to Fig.~\mbox{\ref{fig:sensitivity_sim}}, the presence of either CgLp or CR~CgLp reduces the peak of sensitivity significantly, which is logical because both of them increase the phase margin of the system. At the same time because $\omega_{r}$ is tuned to reduce the higher-order harmonics, it was expected that sensitivity of CgLp and CR~CgLp, namely, $S_\text{CgLp}$ and $S_\text{CR~CgLp}$, closely match the sensitivity of the BLS and the sensitivity approximated by DF, i.e., $S_\text{BLS}$ and $S_\text{DF}$. However, the CR CgLp because of lower higher-order harmonics has closer to ideal behaviour than CgLp. This analysis indicates that the significant improvement in transient behaviour of the CR~CgLp architecture not only has almost no negative effect on steady-state behaviour but also positively affects it by reducing the peak of sensitivity.  \\
To summarize the rule of thumb tuning guideline to CR~CgLp elements the suggested values for different parameters are presented in Table~\ref{tab:thumb}.\\
\begin{table}[t!]
	\centering
	\caption{The rule of thumb tuning values for parameters of CR CgLp.}
	\label{tab:thumb}
	\begin{tabular}{@{}cccccc@{}}
		\toprule
		Parameter & $\omega_r$                & PM                     & $\omega_l$                   & $\omega_h$   & $\omega_f$    \\ \midrule
		Value     & $[\omega_c, 1.5\omega_c]$ & $[15^\circ, 25^\circ]$ & $[0.3\omega_c, 0.6\omega_c]$ & $20\omega_c$ & $20\omega_c$ \\ \bottomrule
	\end{tabular}
\end{table}
The data-based analysis done in previous sections was for mass plants. However, the concepts and the procedure can be done for generalized for  mass-spring-damper plants and the suggested rule of thumb tuning values roughly stands for every mass-spring-damper plant. To verify, in the next section, a practical example CR CgLp will be designed and tested for a precision motion setup which has a mass-spring-damper plant with high-frequency modes.\\
\section{Illustrative Practical Example}
\label{sec:practice}
In order to validate the results of previous sections in precision motion control, an illustrative practical example is presented in this section. Comparison between different controllers such as PID, PID+CgLp and PID+CR~CgLp is presented in this section. For the sake of conciseness, in the rest of the paper, PID+CgLp and PID+CR~CgLp are shortly called, CgLp and CR~CgLp controllers, respectively.
\subsection{Plant}
\begin{figure}[t!]
	\centering
	\includegraphics[width=0.6\textwidth]{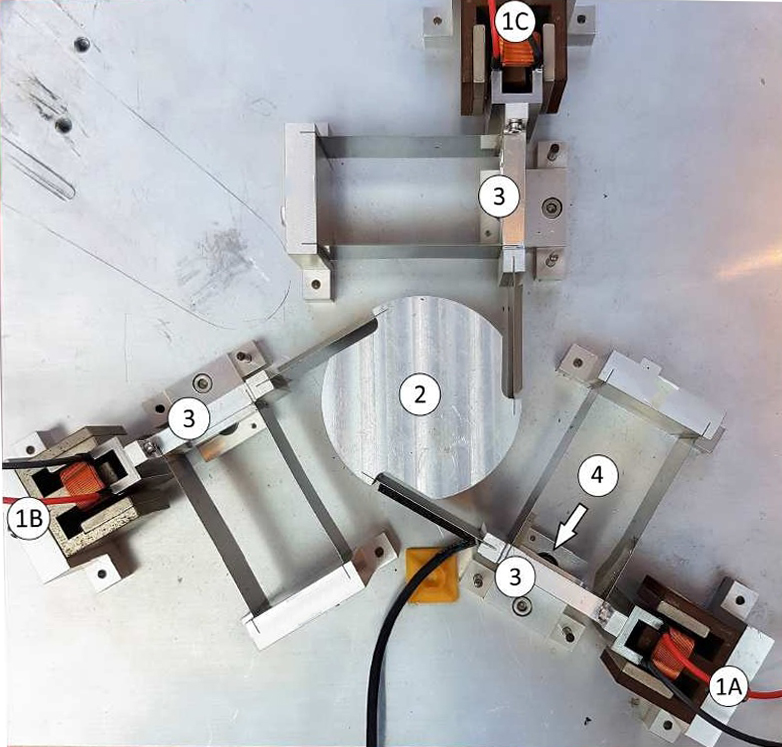}
	\caption{Three degrees of freedom planar precision positioning system called ``Spyder''. Spyder is actuated using three voice coil actuators indicated as 1A, 1B and 1C. The actuators are directly connected to masses indicated by 3. Each of these masses are solely connected to the base through two leaf flexures. The position of these masses are being sensed by linear encoders indicated by 4. }
	\label{fig:spider}
\end{figure}
The precision positioning stage ``Spyder'' is depicted in Fig.~\ref{fig:spider} is a 3 degrees of freedom planar
positioning stage which is used for validation. Since reset controllers in this paper are defined for SISO systems, only the actuator 1A is used to position the mass
rigidly connected to it. An NI compactRIO system which is enhanced by a FPGA is used to implement the controllers at a sampling frequency of 10~kHz. Linear current source power amplifier is used
to drive the voice coil actuator and a Mercury M2000 linear encoder, indicated as 4 in the Fig.~\ref{fig:spider} senses the position of the mass with a resolution of 100 nm.
The FRF of the stage is identified and depicted in Fig.~\ref{fig:indent}. The identification reveals that the plant shows a behaviour similar to that of a collocated double mass-spring-damper with additional parasitic dynamics at high frequencies. For the sake of better illustration of control design a mass-spring-damper transfer function has been fitted to the FRF data presented in the Eq.~\eqref{eq:plant}.
\begin{equation}
	\label{eq:plant}
	P(s)=\frac{9836e^{-0.0001s}}{s^2+8.737s+7376}
\end{equation}
\begin{figure}[t!]
	\centering
	\includegraphics[width=\textwidth]{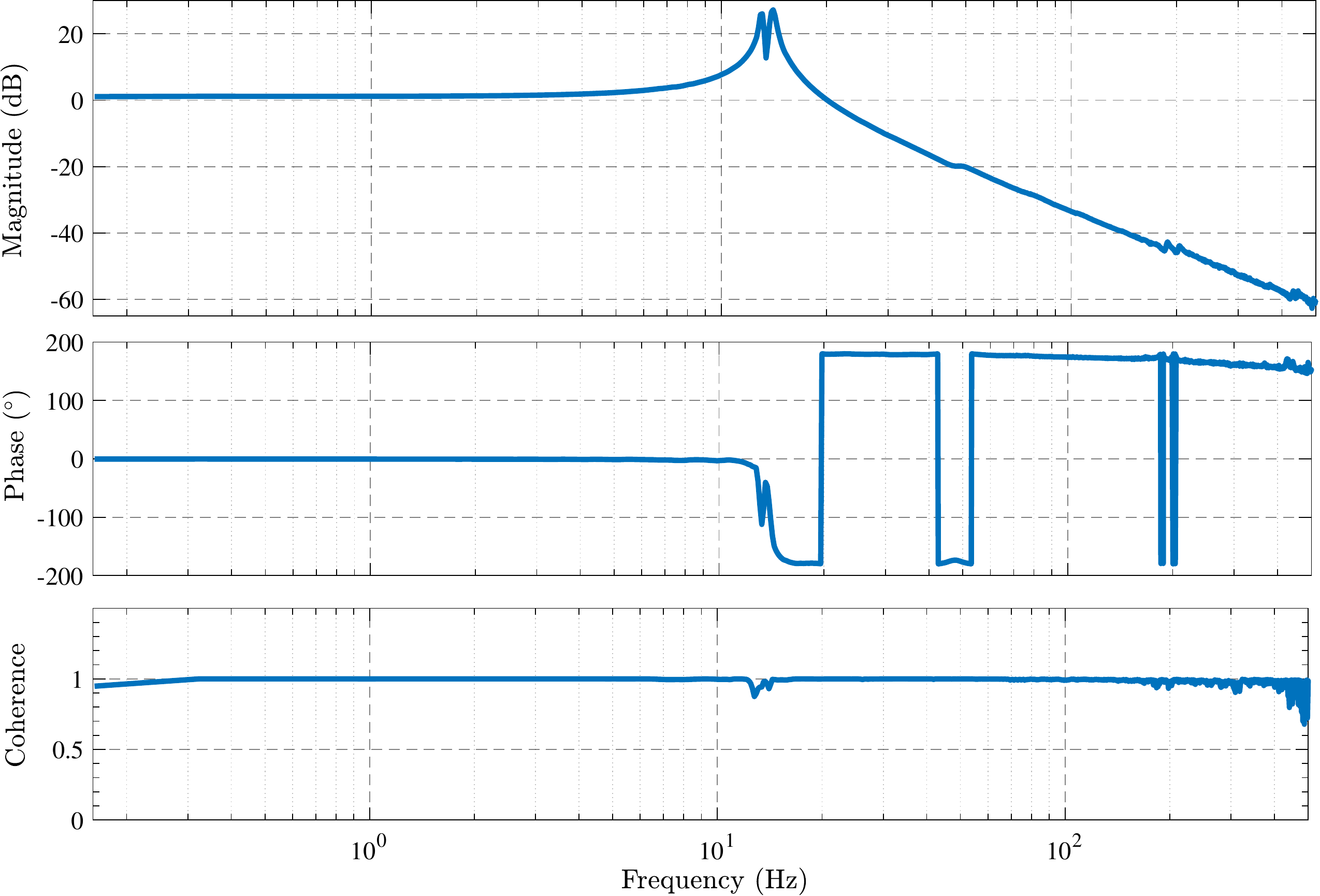}
	\caption{FRF identification of actuator 1A positioning the attached mass.}
	\label{fig:indent}
\end{figure}
\subsection{Controller Design Approach}
In order to compare the performance of PID and CR~CgLp and show the superiority of the CR~CgLp over PID in both steady-state and transient, four controllers were designed. PID controllers are tuned following the tuning rules presented in~\mbox{\cite{schmidt2020design}} and reset controllers are designed following the guidelines presented in the paper.  The controller loop is already depicted in Fig~\ref{fig:block_closed}. However, due to presence of noise in practice, a first order low-pass filter, $\frac{1}{s/\omega_z+1}$, has been added to the loop. The parameters for designed controllers is presented in Table~\ref{tab:c_parameters}.\\
\begin{table}[t!]
	\caption{The parameters for designed controllers. $\omega_c=400$ Hz.}
	\label{tab:c_parameters}
	\centering
	\begin{adjustbox}{max width=\textwidth}
		\begin{tabular}{@{}ccccccccc@{}}
			\toprule
			Parameter & $\omega_i$    & $\omega_d$     & $\omega_t$    & $\omega_z$  & $\omega_r$ & $\omega_l$    & $\omega_h$  & $\omega_f$    \\ \midrule
			PID \#1   & $\omega_c/10$ & $\omega_c/2.5$ & $2.5\omega_c$ & $5\omega_c$ & N/A        & N/A           & N/A         & N/A          \\
			PID \#2   & $\omega_c/10$ & $\omega_c/5$   & $5\omega_c$   & $5\omega_c$ & N/A        & N/A           & N/A         & N/A          \\
			PID \#1 + CgLp   & $\omega_c/10$ & $\omega_c/2.5$   & $2.5\omega_c$   & $5\omega_c$ & $\omega_c$ 	  & N/A           & N/A         & $20\omega_c$ \\ 			
			PID \#1 + CR CgLp   & $\omega_c/10$ & $\omega_c/2.5$   & $2.5\omega_c$   & $5\omega_c$ & $\omega_c$ & $\omega_c/8$ & $5\omega_c$ & $20\omega_c$ \\ \bottomrule
		\end{tabular}
	\end{adjustbox}	
\end{table}
Since the input signal to $L(s)$ is $e(t)$, this element will amplify the noise present in $e(t)$ and thus creates excessive zero crossings and thus excessive reset actions~\cite{cai2020optimal}. In order to avoid this phenomenon, $\omega_{h}$ has chosen to be smaller than the rule-of-thumb guidelines provided in previous sections to better attenuate the high-frequency content of the signal. This change in $\omega_h$ increases the overshoot in step response, to compensate, $\omega_{l}$ has chosen to be smaller than rule-of-thumb guidelines.  \\
PID \#1 can also be considered the BLS for the CR~CgLp \emph{controller}, since the latter is simply PID \#1 with CR~CgLp \emph{element} preceding it, as can be seen in Fig.~\ref{fig:block_closed}. The practical study will show that adding the CR~CgLp \emph{element} to a linear PID controller will improve the transient and the steady-state characteristics simultaneously.\\
The open-loop HOSIDF analysis of the CR~CgLp controller and the bode plot of the PID controllers are depicted in Fig~\ref{fig:hosidf_open}. Due to choosing of $\omega_{r}$ according to Fig.~\ref{fig:wr_sweep} and the architecture of CR~CgLp, it can be seen that the magnitude of higher-order harmonics for CR CgLp are at least 60 dB smaller than first-order harmonic. Thus, it is expected that the steady-state response of the system closely follows the amplitude of the first-order harmonic. 
\begin{figure}[t!]
	\centering
	\includegraphics[width=\textwidth]{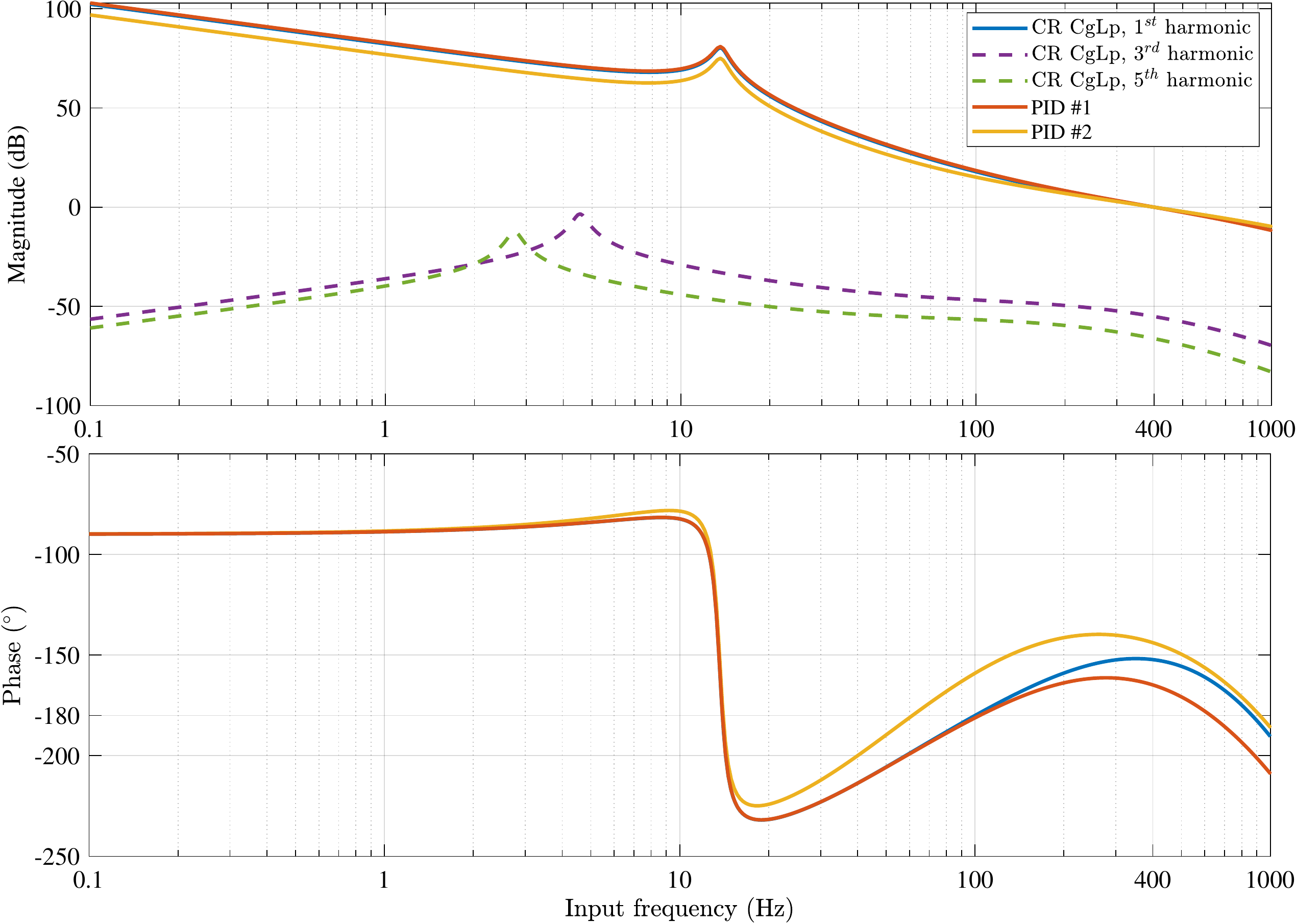}
	\caption{Open-loop HOSIDF analysis of the CR CgLp and Bode plot of PID controllers including the plant. PM for the CR CgLp, PID \#1 and PID \#2 are respectively, $25^\circ$, $15^\circ$ and $35^\circ$.}
	\label{fig:hosidf_open}
\end{figure}
\subsection{Comparison of the steady-state response}
For comparison of the precision of the controllers in terms of steady-state sinusoidal tracking the sensitivity plot of the controllers are depicted in Fig.~\ref{fig:hosidf_open}. For this purpose, sinusoidal signals between 1 and 500 Hz has been input as $r(t)$ and $\frac{\Vert e(t) \Vert_2}{\Vert r(t) \Vert_2}$ has been calculated and plotted for each sinusoidal.\\
In the range of $[1,10]$ Hz, the sensitivity of all controllers seemed to be lower bounded by -60 dB, this effect is caused by the quantization and the precision of the sensor.  However, comparing PID \#1 and CR CgLp in range of $[10, 500]$ Hz reveals that performance of the CR CgLp closely matches PID \#1 in lower frequencies and its peak of sensitivity is 1.5 dB lower. Thus, one can conclude that the steady-state performance of the linear controller is improved by introducing the proposed element. For the case of PID \#2, the clear waterbed effect can be seen, i.e., by widening the band of differentiation, at the cost of losing precision at lower frequencies, the peak of sensitivity is reduced. As opposed by CR CgLp, where reduction of peak of sensitivity achieved without sacrificing the precision at lower frequencies. Although in linear control context one would expect PID \#2 to have better transient response because of lower speak of sensitivity, in the next subsection, it will be shown that this does not hold true in nonlinear context and CR~CgLp controller shows better transient response despite of having higher peak of sensitivity.
\begin{figure}[t!]
	\centering
	\includegraphics[width=\textwidth]{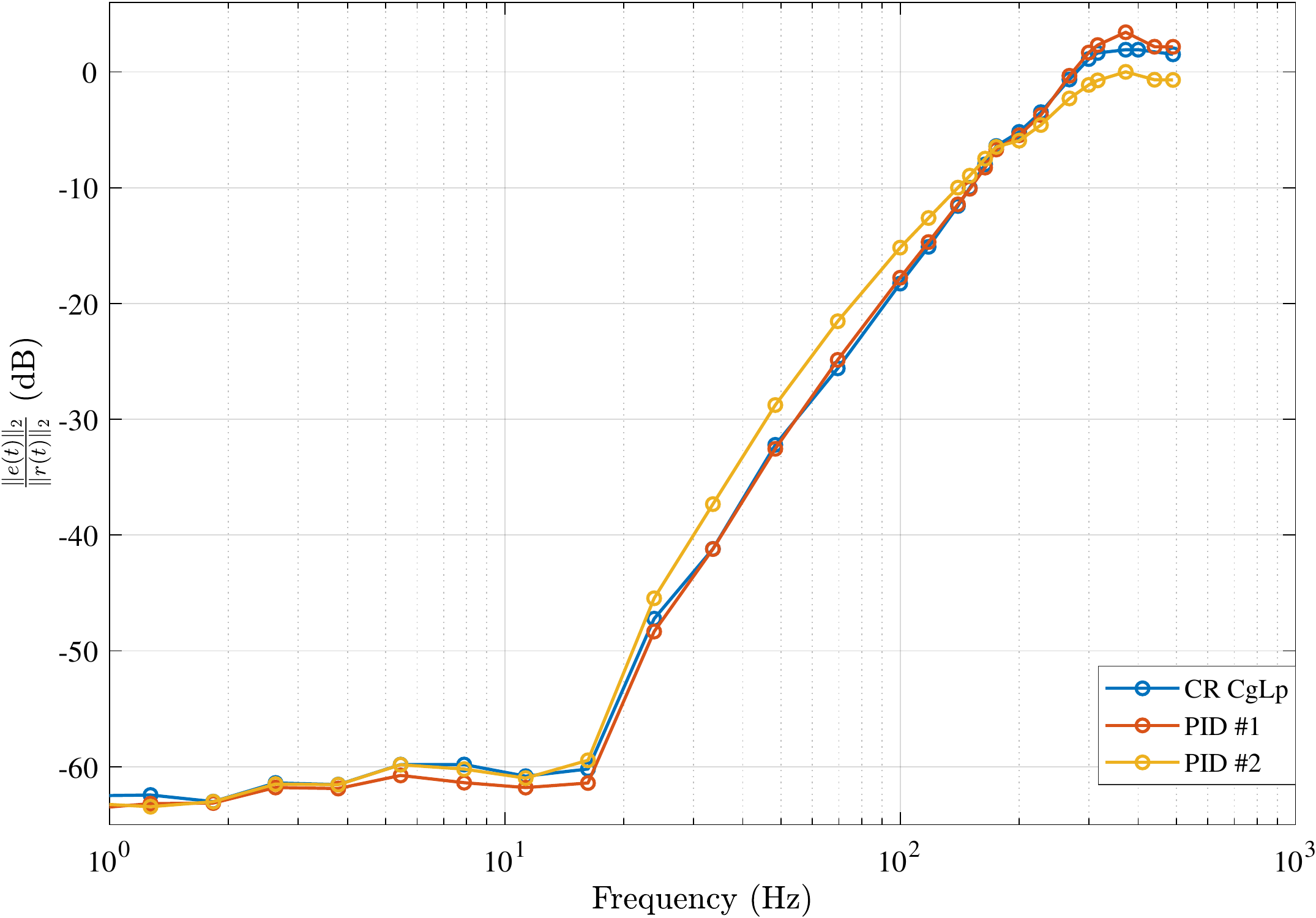}
	\caption{The closed-loop sensitivity of controllers for sinusoidal signals with frequencies in $[1,500]$ Hz. Frequencies above 500 Hz are not recorded due to the actuator limitations. The sensitivity plot of CgLp closely matches that of CR~CgLp, thus, it is not shown for the sake of clarity.}
	\label{fig:sensitivity_practice}
\end{figure}
\subsection{Comparison of the transient response}
\begin{figure}[t!]
	\centering
	\includegraphics[width=\textwidth]{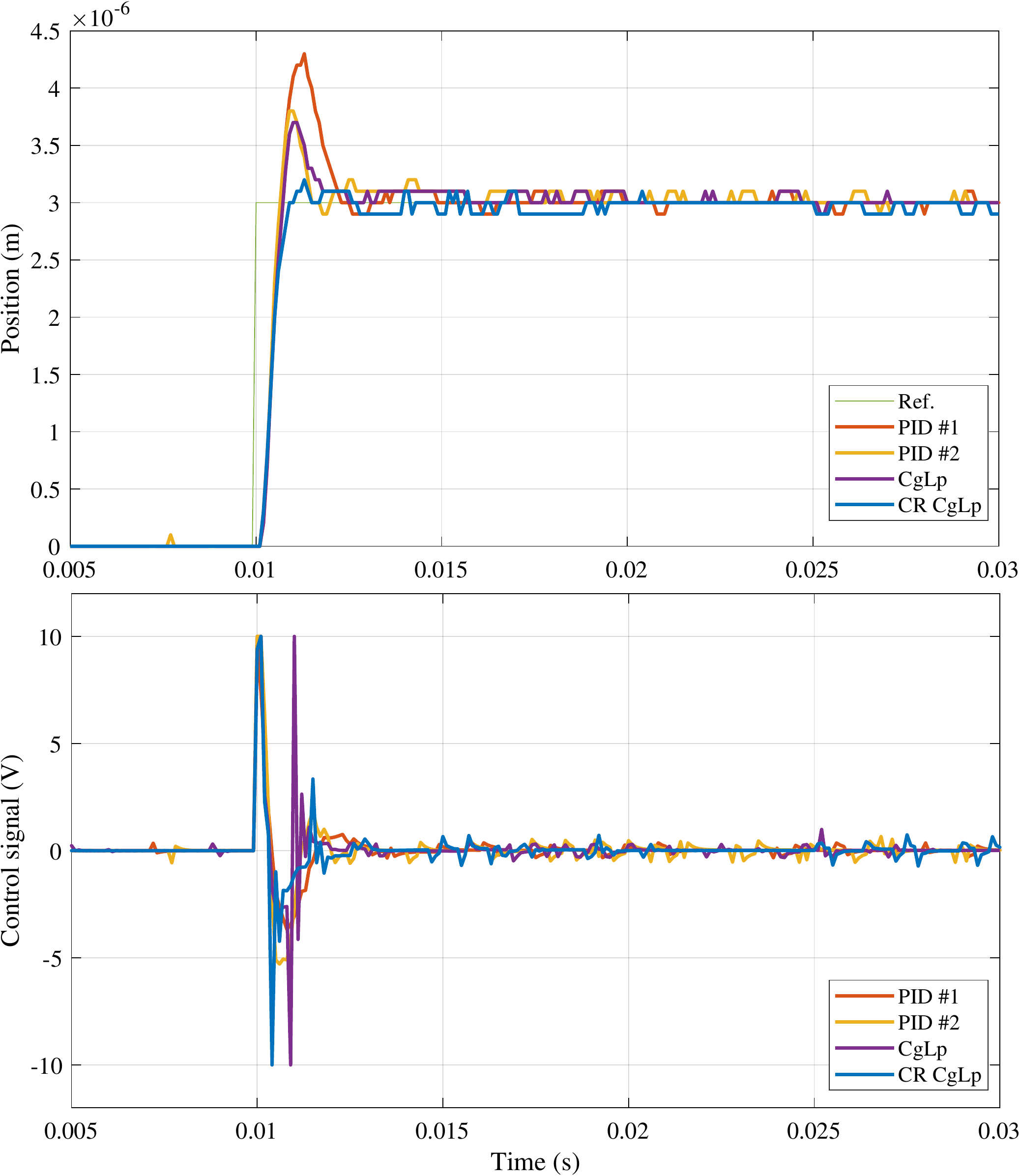}
	\caption{Step response and its corresponding control signal for the controllers introduced in Table~\ref{tab:c_parameters}. The overshoot for CR CgLp, PID~\#1, PID~\#2 and CgLp is respectively, $3\%$, $43\%$, $28\%$, $22\%$ and the $96\%$ settling times are respectively, 11.6 ms, 12.2 ms, 14.4 ms and 11.8 ms.}
	\label{fig:step_practice}
\end{figure}
For comparison of the step responses of the controllers, a step input of $0.15\mu$m height has been used. The response of the controllers are depicted in Fig.~\ref{fig:step_practice}. As it can be seen the CR CgLp shows a no-overshoot performance where PID \#1 shows an overshoot of $38\%$. It is noteworthy that according to Fig.~\ref{fig:sensitivity_practice}, these two controllers have matching sensitivity at lower frequencies. The settling time has also improved by $25\%$. This example clearly demonstrates that by adding CR~CgLp \textit{element} to an existing PID linear loop, one can achieve a no-overshoot performance and generally  significantly improved transient response while maintaining the steady-state precision. \\
The peak of sensitivity for both CgLp and CR~CgLp controllers are the same, however the overshoot of the CR~CgLp is $28\%$ lower than CgLp and that of CgLp is $10\%$ lower than that of PID. This results validates that the transient performance of the reset controllers, especially the overshoot, is affected but not solely by PM and peak of sensitivity. The architecture and $\omega_{l}$ also play role. The effect of $\omega_{l}$ will be validated further.\\
The reduction of overshoot for PID~\#2 compared to PID~\#1 was obvious due to wider band of differentiation and thus reduced peak of sensitivity. However, despite the fact that its peak of sensitivity is lower than CR~CgLp, the overshoot is still larger than that of CR~CgLp. Meanwhile steady-state precision was already shown to be lower that CR~CgLp. \hl{It has to be noted, because of the relatively high bandwidth which is chosen for the controllers, i.e., 400 Hz, and limitations of the actuator, control signal for the controllers come close to saturation in only one sample of time. Nevertheless, it will be shown later that this is not the case for lower bandwidths, even for larger references.} 
\subsection{The effect of $\omega_l$}
\begin{figure}[t!]
	\centering
	\includegraphics[width=\textwidth]{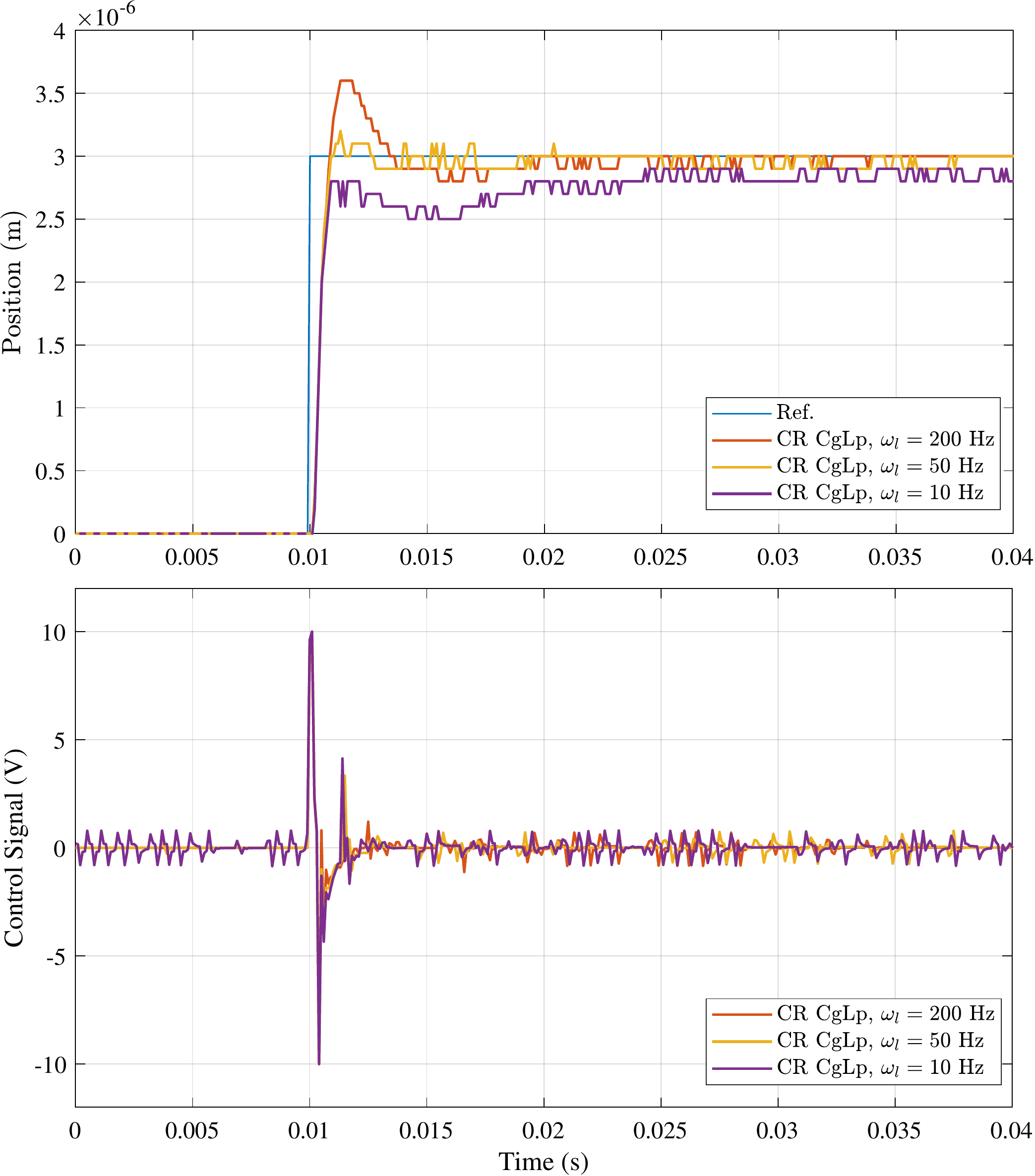}
	\caption{The effect of varying $\omega_{l}$ on transient response and its corresponding control signal of the CR CgLp.}
	\label{fig:step_vary_wl}
\end{figure}
In Fig.~\ref{fig:step_practice}, $\omega_{l}=50$ Hz. In order to validate the effect of $\omega_{l}$ on transient response, the step response for different values of $\omega_{l}$ while maintaining the other parameters is depicted in Fig.~\ref{fig:step_vary_wl}. It can be clearly seen that overshoot keeps decreasing with reduction of $\omega_{l}/\omega_{c}$. Furthermore, it can be also validated that settling time will increase when $\omega_{l}/\omega_{c}$ drops below a certain threshold. This phenomenon can be due to the fact that too much reset and resetting too soon can jeopardize the effect of integrator.  It is noteworthy that according to Proposition~\ref{rem:df}, the value of $\omega_{l}$ does not have an effect on DF and thus steady-state tracking performance of the system.
\subsection{Complex-order behaviour}
\begin{figure}[t!]
	\centering
	\includegraphics[width=\textwidth]{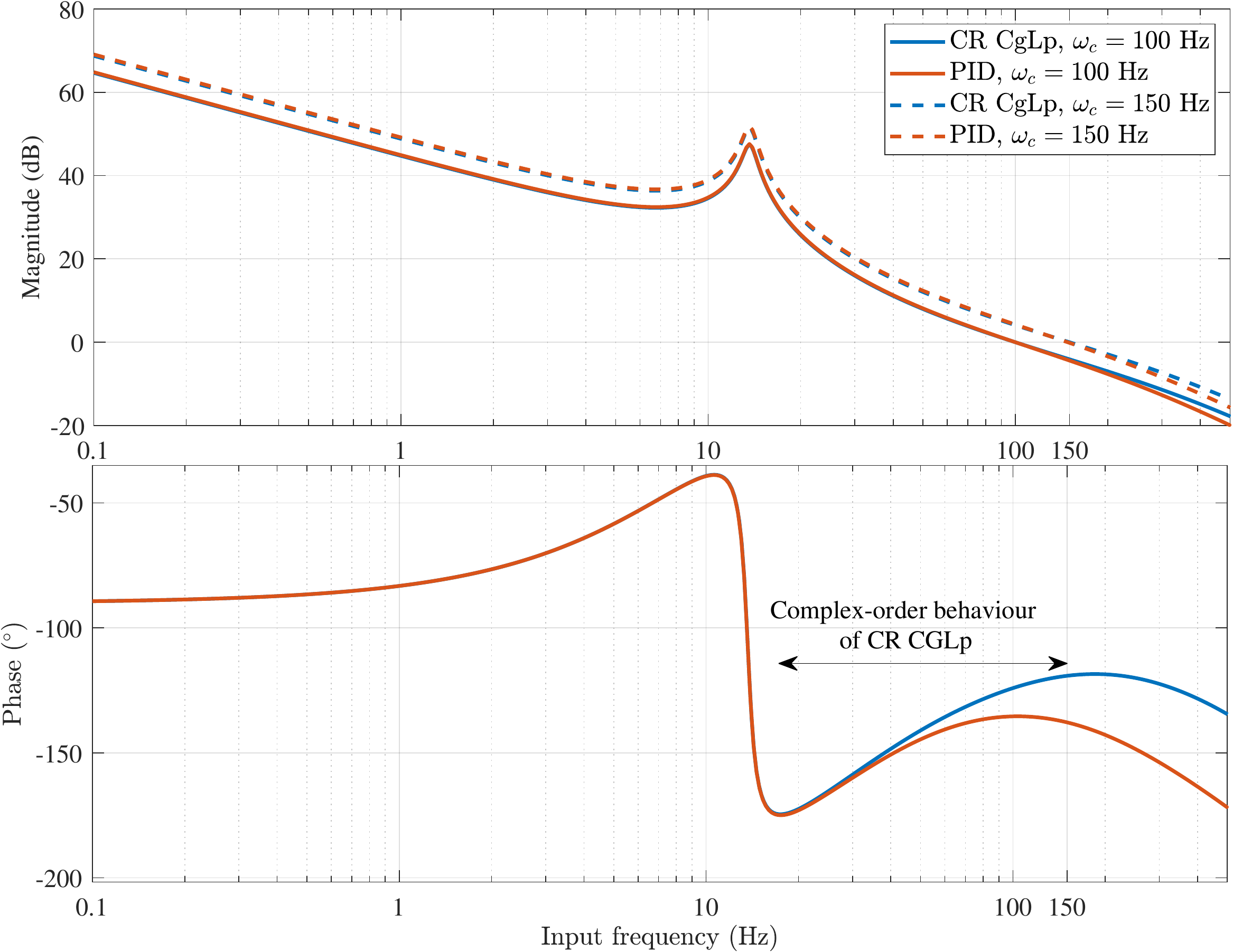}
	\caption{Bode diagram for a PID and DF diagram of a CR CgLp, showing the complex-order behaviour of CR CgLp. PM for CR~CgLp at 100 Hz is $55^\circ$ and for 150 Hz is $60^\circ$. PM for the PID at 100 Hz is $45^\circ$ and for 150~Hz is $43^\circ$. }
	\label{fig:bode_complex}
\end{figure}
Another interesting behaviour of the CR CgLp controller is the ability to create a complex-order behaviour as depicted in Fig.~\ref{fig:bode_complex}. Two controllers have been designed for $\omega_c= 100$~Hz. In the case of gain variation of 5 dB, $\omega_{c}$ will change to 150~Hz, in such a situation, PID loses $3^\circ$ of PM while CR~CgLp will show a complex-order behaviour, meaning the phase increases while gain decreases~\cite{valerio2019reset}, and gain $5^\circ$ more PM. Thus the modulus margin for PID is expected to be decreased and for CR~CgLp to be increased.\\
Furthermore, an increase on overshoot of PID and a decrease in that of CR~CgLp is expected. The validation of this expectation has been been done in practice and the step responses are shown in Fig~\ref{fig:step_complex}. However, the increase in PM is not the only reason for decrease of overshoot in CR~CgLp. Since $\omega_{c}$ is increased and $\omega_{l}$ has been kept constant, the ratio of $\omega_{l}/\omega_{c}$ is subsequently reduced, which also helps the reduction of overshoot.  
\begin{figure}[t!]
	\centering
	\includegraphics[width=\textwidth]{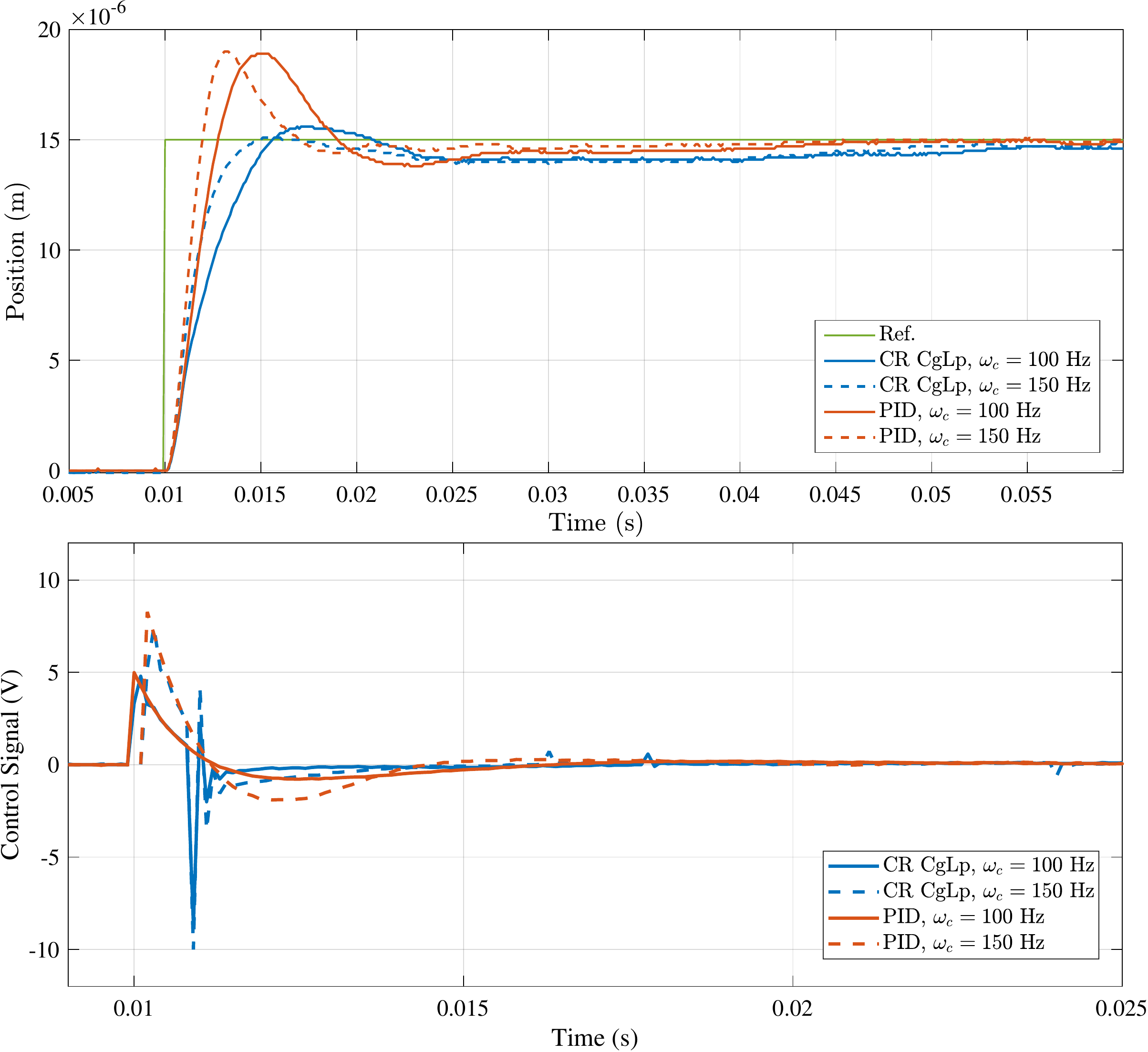}
	\caption{The step response and its corresponding control signal for the controllers shown in Fig.~\ref{fig:bode_complex} for gain variation of 5 dB. PM for PID increases from $25\%$ to $26\%$ for increased $\omega_c$ while PM for CR~CgLp decreases from $4\%$ to $0\%$. }
	\label{fig:step_complex}
\end{figure}

\section{Conclusion}
A new architecture for reset elements, named ``Continuous Reset Element'' was presented in this paper. Such an architecture consists of having a linear lead and lag element, before and after of a reset element. It was shown that such an architecture not only does not influence the DF gain and phase of reset elements, but also reduces the magnitude of higher-order harmonics, which will positively effect the steady-state tracking precision of the reset controllers. Furthermore, it was shown that having a strictly proper lag element after the reset element will make the output of the  reset element continuous as opposed to conventional reset elements.\\
Moreover, it was shown that such CR architecture also can significantly improve the transient response of the reset control systems for mass-spring-damper plants without negatively affecting the steady-state performance, an overcoming over waterbed effect. To this end, a numerical study was done on a reset element called CgLp and a mass plant and it was shown that by using the CR architecture, the settling time and overshoot of the CR CgLp control system can be improved both comparing to CgLp element itself and the BLS.\\
To further validate the achieved results, a practical example was introduced where a precision motion setup was identified and four controllers were implemented and compared in terms of transient and steady-state performance. It was shown that for a mass-spring-damper plant the CR~CgLp controller was able to achieve a no-overshoot performance and a reduced settling time while matching the steady-state performance of the PID BLS at lower frequencies and a showing a reduced peak of sensitivity.\\
However, the presence of a lead element before a reset element can introduce excessive reset actions to the control because of noise. To avoid such a phenomenon a low-pass filter or in general term a shaping filter can be used to remove the high-frequency content of the signal. For which a more extensive research is required. The latter can be ongoing work of the propose design.

% or


%\appendices
%% you can choose not to have a title for an appendix
%% if you want by leaving the argument blank
%\section{}

\section*{Acknowledgment}

This work was supported by NWO, through OTP TTW project \#16335.

%\section*{References}

\bibliography{ref}

\end{document}